\documentclass{article}       
\usepackage{authblk}
\usepackage{times}
\usepackage{amsmath,amsfonts,amssymb,amsthm}
\usepackage{graphicx}
\usepackage{color, colortbl}
\definecolor{Gray}{gray}{.9}
\usepackage{longtable}

\usepackage{algpseudocode}
\usepackage{algorithm}
\usepackage{xspace}
\usepackage{afterpage}
\usepackage{enumitem}
\setlist[description]{leftmargin=6pt,labelindent=6pt}

\newtheorem{lemma}{Lemma}
\newtheorem{theorem}{Theorem}
\newtheorem{corollary}{Corollary}

\newtheorem{definition}{Definition}
\newtheorem{example}{Example}

\usepackage{mathtools}

\newcommand{\Athena}{\textsc{Athena}\xspace}
\newcommand{\Metis}{\textsc{Metis}\xspace}
\newcommand{\Minerva}{\textsc{Minerva}\xspace}
\newcommand{\B}{{{B2}}\xspace}
\newcommand{\R}{{{R2}}\xspace}
\newcommand{\Bravo}{\textsc{Bravo}\xspace}

\begin{document}              
\pagenumbering{arabic} \baselineskip 15pt \date{}
\pagestyle{myheadings} \markboth{\footnotesize}{\footnotesize \Athena Ballot Polling Audits}
 \title{The \Athena Class of Risk-Limiting Ballot Polling Audits\thanks{This material is based upon work supported in part by NSF Awards 2015253 and 1421373}}
\author[1]{Filip Zag\'{o}rski\thanks{filip.zagorski@gmail.com, Author was partially supported by
Polish National Science Centre contract number DEC-2013/09/D/ST6/03927}}

\author[2]{Grant McClearn}
\author[2]{Sarah Morin}
\author[ ]{Neal McBurnett}
\author[2]{Poorvi  L. Vora\thanks{poorvi@gwu.edu}}

\affil[1]{Wroclaw University of Science and Technology}
\affil[2]{Department of Computer Science\\The George Washington University}

\date{\today}       
\maketitle                    
\begin{abstract}
The main risk-limiting ballot polling audit in use today, \Bravo, is designed for use when single ballots are drawn at random and a decision regarding whether to stop the audit or draw another ballot is taken after each ballot draw ({\em ballot-by-ballot (B2)} audits). On the other hand, real ballot polling audits draw many ballots in a single round before determining whether to stop ({\em round-by-round (R2)} audits). We show that \Bravo results in significant inefficiency when directly applied to real \R audits. We present the \Athena class of \R stopping rules, which we show are risk-limiting if the round schedule is pre-determined (before the audit begins). We prove that each rule is at least as efficient as the corresponding \Bravo stopping rule applied at the end of the round. We have open-source software libraries implementing most of our results. 

We show that \Athena halves the number of ballots required, for all state margins in the 2016 US Presidential election and a first round with $90\%$ stopping probability, when compared to \Bravo (stopping rule applied at the end of the round). We present simulation results supporting the $90\%$ stopping probability claims and our claims for the risk accrued in the first round. Further, \Athena reduces the number of ballots by more than a quarter for low margins, when compared to the \Bravo stopping rule applied on ballots in selection order. This implies that keeping track of the order when drawing ballots \R is not beneficial, because \Athena is more efficient even without information on selection order. These results are significant because current approaches to real ballot polling election audits use the \B \Bravo rules, requiring about twice as much work on the part of election officials. Applying the rules in selection order requires fewer ballots, but keeping track of the order, and entering it into audit software, adds to the effort.  

All our contributions are for audits with zero error of the second kind. Our approach relies on analytical expressions we derive for stopping probabilities. The results of these analytical expressions are verified by comparison with the percentiles Lindeman {\em et al.} previously obtained using \B \Bravo simulations \cite{bravo}. 

We believe our results may be applied in a straightforward fashion to other SPRTs with $\beta=0$ if the stopping condition is monotonic increasing with the number of winner ballots (Bayesian audits are an example) but proofs in this paper apply only to \Bravo. 

\end{abstract}
\section{Introduction}
\label{sec:intro}
The most popular examples of election tabulation ballot polling audits include \Bravo \cite{bravo} and Bayesian audits \cite{bayesian-audits}; these audits may be viewed as special cases/extensions of the {\em sequential probability ratio test} ({\em SPRT}), see \cite{Bayesian-RLA,RLA-2020}. When the decisions of whether to stop the audit or draw more ballots are taken after each ballot draw, and the stopping condition is satisfied exactly when the audit is stopped, these audits---as {\em SPRT}s---are {\em most efficient} audits. The term {\em most efficient} refers here, as elsewhere, to an audit requiring the smallest expected number of ballots given either hypothesis: a correct election outcome or an incorrect one, if the election is drawn from the assumed prior. The expectation is taken over the randomness of the ballot draws, and, in the case of Bayesian audits, also the randomness of the true tally (Bayesian audits treat the true tally as an unknown random variable). 

In real election audits, multiple ballots are drawn in a round before a decision is taken. This paper shows that \Bravo is not a most efficient test in this case, and proposes the \Athena class of more efficient tests, demonstrating significant decreases in first-round sizes, and proving that the tests are risk-limited if the round schedule is pre-determined (before the audit begins). This could be of consequence for election audits of the 2020 US Presidential election. 

\subsection{Problem}
We refer to audits where decisions are taken after each ballot draw as {\em ballot-by-ballot} or {\em B2} audits. The general audit, however, is a {\em round-by-round} or \R audit where, in the $j^{th}$ round, some ballots are drawn, after which a decision is taken regarding whether to (a) stop the audit and declare the election outcome correct, (b) stop the audit and go to a manual recount, or (c) draw the $(j+1)^{th}$ round. A \B audit is a special case of the \R audit, when a single ballot is drawn in each round. 

There are two ways to apply \B audit rules to an \R audit. Consider a total of $n_j$ ballots drawn after the $j^{th}$ round, of which $k_j$ are for the reported winner. 
\begin{itemize}
\item {\em End-of-round:} In this application, the \B stopping rule for $k_j$ winner ballots in a sample of $n_j$ ballots determines whether the audit will stop. 
\item {\em Selection-ordered-ballots:} In this application, ballot order is recorded and the \B stopping condition is tested $\forall n \leq n_j$. The audit stops if the \B condition is satisfied for any value of $n \leq n_j$. 
\end{itemize}
{\em Selection-ordered-ballots} is generally more efficient than {\em end-of-round} as a means of applying \B rules to \R audits, but requires the significant additional effort of preserving enough information to be able to recreate the subtotals of winner ballots in selection order. {\em End-of-round} relies only on the tallies and does not require selection order. As our paper shows, neither is a most efficient \R stopping rule. 

One may view the problem we address as lying somewhere between (a) the problem solved by Neyman-Pearson \cite{van-trees}: derive a single-use binary hypothesis testing rule satisfying certain error criteria, and (b) the problem solved by Wald \cite{wald}: derive a stopping condition for sequential sampling, satisfying certain error criteria, where the condition is tested draw-by-draw. We address the problem of sequential sampling in rounds, where the condition is tested after multiple draws. 

\subsection{Our Contributions}
Our contributions are as follows:
\begin{enumerate}
\item We derive analytical expressions for the risk and probability of stopping, given the history of rounds and the margin for the \Bravo audit. Treating the \B \Bravo audit as an \R audit with $n_j=j$, we verify that the expressions we derive predict the stopping percentiles originally obtained by Lindeman {\em et al.} using \Bravo simulations \cite[Table 1]{bravo}. The average of the absolute value of the fractional difference between our results and those of \cite{bravo} is 0.13\%. The largest difference has value $190$ ballots, corresponding to a fractional difference of 0.41 \%, in the estimate of the expected number of ballots drawn for a margin of 1\%. This difference could be due to small inaccuracies in our computational approach (such as rounding off errors or the finiteness of summations involved in the computations) or the finiteness of the number of simulations used to generate the results of \cite{bravo}. Our approach is easily extended to audits with stopping conditions that are monotone increasing in the number of ballots for the announced winner, such as Bayesian audits. The code for computing these expressions is available as a MATLAB library, released as open-source under the MIT License \cite{brla_explore}. 
\item We present the \Athena class of \R stopping rules for audits \Minerva and \Athena and prove that, if the round schedule is pre-determined (before the audit begins), \Minerva and \Athena are both risk-limiting and at least as efficient as the corresponding {\em end-of-round} \Bravo stopping rule. Another audit from the \Athena class, \Metis, is out of scope for this draft. 
\item We provide experimental results and software to support the use of the proposed audits: 
\begin{itemize}
\item To illustrate the efficiency improvements, we compute (without simulations, using the derived analytical expressions), for each state in the 2016 US Presidential election, risk limit $\alpha =0.1$ and a stopping probability of 0.9, first round sizes for {\em end-of-round} \Bravo and \Athena. We find that \Athena requires about half the number of ballots, across all margins. 
\item We compute first round sizes for {\em selection-ordered-ballots} \Bravo and find that \Athena requires about $15-29\%$ fewer ballots for the data of the 2016 US Presidential election, with the improvement being better for smaller margins. Thus \Athena is more efficient than {\em selection-ordered-ballots} \Bravo and does not require the additional book-keeping of recording selection ballot order. 
\item We present the results of simulations supporting our predictions of first round stopping probabilities and the risk-limiting properties of \Athena.
\item Our code for the audits is available as MATLAB and Python libraries \cite{Athena,brla_explore,r2b2}. All code is released as open source under the MIT license.
\end{itemize}
This contribution is important because a number of states have undertaken ballot polling pilots in the last year and plan to use ballot polling audits in November 2020. We hope that our results and code can help developers of auditing software. We also note that, in many scenarios, ballot comparison or batch comparison audits could be more feasible. One may also consider combinations of ballot comparison and ballot polling audits, such as described in \cite{suite}.  
\item The \Athena class of \R stopping rules is a class of \B rules when round size is one. Of theoretical interest, we prove that \B \Minerva (round size one) has the same stopping rule as \B \Bravo, as does \B \Athena for some values of its parameters. 
\end{enumerate}
We do not claim that the audits of the \Athena class are the most efficient \R audits with zero error of the second kind ($\beta=0$). The problem of finding the most efficient \R audits is open. 

Unlike the {\em SPRT} and other Martingale-based approaches, the stopping rules for audits of the \Athena class use information about the history of round sizes. For this reason the stopping condition for rounds other than the first one does not depend only on the cumulative sample size and number of winner ballots drawn, but also on the history of individual round sizes. 

We do not address some simple extensions of our work in this paper. For example, we do not consider audits with a limit on the total number of ballots drawn in the polling audit (if the audit fails to stop, a full sequential hand count would follow). Were we to do so, we could provide audits with larger stopping probabilities given the same risk limit. 

This paper focuses exclusively on \Bravo. However, we expect that our results extend to other risk-limiting {\em SPRT}s with $\beta=0$ and a stopping condition that is monotonic in the number of winner ballots drawn; examples include Bayesian audits. 

\subsection{Organization}
Section \ref{sec:model} presents the model and related work. Section \ref{sec:problem} motivates the problem with an example demonstrating that the application of \B rules to an \R audit results in inefficiencies. Section \ref{sec:informal} introduces the \Athena class of audits with examples and provides insight into why the audits are risk-limiting and more efficient than either \R application of \B \Bravo. Section \ref{sec:computing} describes the analytical approaches for computing probabilities for multiple-round audits. Section \ref{sec:algorithms} presents the \Minerva and \Athena audits, and Section \ref{sec:theorems} presents rigorous claims of their risk-limiting and efficiency properties in the form of Theorems and Lemmas. Section \ref{sec:results} presents experimental results. Section \ref{sec:conc} concludes.  Proofs are in the Appendix. 

\subsection{Acknowledgements}
We gratefully acknowledge comments on an early draft by: Matthew Bernhard, Amanda Glazer, Mark Lindeman, Jake Spertus, Mayuri Sridhar, Philip B. Stark, Damjan Vukcevic.  

This version is updated from the previous one to reflect only a couple, and not all, of their valuable suggestions. Importantly, this draft notes that the risk-limiting property is proven only when round sizes are pre-determined. We have also made some of the changes, including some improvements to our notation, and the inclusion of the number of distinct ballots in tables listing the number of ballots required for a first round with 90\% stopping probability for 2016 tallies. We plan to soon further update the manuscript to include estimated first round sizes for the 2021 tallies and address all other suggestions made by the reviewers. 

This research was sponsored in part by NSF Awards 2015253 and 1421373. 
\section{The Model}
\label{sec:model}
\noindent
We consider a plurality contest and assume ballots are drawn with replacement. We assume all ballots have a vote for either the winner or the loser; because ballots are sampled with replacement, our argument is easily extended to contests with multiple candidates and invalid ballots (as for \Bravo, for example, see \cite{RLA}). We denote by $w$ the true winner, $w_a$ the announced winner, $\ell _a$ the announced loser and $p$ the announced fractional tally for $w_a$ (typically based on preliminary, uncertified results). 

A polling audit will estimate whether $w_a$ is the true winner. We denote by $n_j$ the total number of ballots drawn at the end of the $j^{th}$ round, and by $k_j$ the corresponding total number of ballots for the winner. Hence the number of new ballots drawn in round $j$ is $n_j - n_{j-1}$, and the number of new votes for the winner drawn in round $j$ is $k_j - k_{j-1}$. If necessary, one may assume that $n_0, k_0=0$. We often refer to $[ n_1, n_2, \ldots, n_j, \ldots]$ as the {\em round schedule}. A \B audit is an \R audit with round size $n_j = j$. That is, the round schedule of a \B audit is $[1, 2, \ldots, j, \ldots]$. 


The total number of ballots drawn at any time during the audit is denoted $n$ (if the number of rounds drawn so far is $j$, $n=n_j$). The random variable representing the number of ballots drawn so far for the winner is represented by 
$K$ (and $K_j$ to represent the number of ballots drawn for the winner up to the $j^{th}$ round). We use $k^*$, $k_*$ and $\tilde{k}$ to represent specific numbers of winner ballots as well. 

The entire sample drawn up to the $j^{th}$ round, in sequence, forms the {\em signal} or the {\em observation}; the corresponding random variable is denoted $X_j$, the specific value $x_j$. The entire sample drawn so far is denoted $X$, its specific value $x$. We do not {\em a priori} assume a last round for the audit. The audit stops when it satisfies the stopping condition. 

\subsection{The Model}
\label{sec:general_audit}
\noindent
We model the audit as a binary hypothesis test: 
\begin{description}
\item Null hypothesis $H_0$: The election outcome is the closest possible incorrect outcome: $w \neq w_a$ and the fractional vote count for $w_a$ is $ \frac{1}{2}$. In particular, if the total number of valid votes is even, the election is a tie. If the total number of valid votes is odd, the margin is one in favor of $\ell _a$. In this case, we assume that the number of valid votes is large enough that the fractional vote count is sufficiently close to $\frac{1}{2}$. Henceforth, we will refer to both cases as being represented by a fractional vote count of $ \frac{1}{2}$. 
\item Alternate hypothesis $H_a$: The election outcome is correct: $w = w_a$ and the fractional vote count is as announced.  
\end{description}

After each round the test $\mathcal{A}$ takes $X$ as input and outputs one of the following: 
\begin{itemize}
\item {\em Correct:} The test estimates that $w=w_a$ and the audit should stop.
\item {\em Incorrect:} The test estimates that $w \neq w_a$. We stop drawing votes and proceed to perform a complete hand count to determine $w$.
\item {\em Undetermined (draw more samples):} We need to draw more ballots to improve the estimate.
\end{itemize}

When the audit stops, it can make one of two kinds of errors:
\begin{enumerate}
\item {\em Miss:} A {\em miss} occurs when $w \neq w_a$ but the audit misses this, and outputs {\em Correct}. We denote by $P_M$ the probability of a miss: 
\[P_M = Pr[\mathcal{A}(X) = \textit{Correct} \mid H_0] \] 
$P_M$ is the {\em risk} in risk limiting audits and the Type I error of the test.
\item {\em Unnecessary Hand Count:} Similarly, if $w=w_a$, but  the audit estimates that a hand count must follow, the hand count is unnecessary. We denote the probability of an {\em unnecessary hand count} by $P_U$: 
\[P_U = Pr[\mathcal{A}(X) = \textit{Incorrect} \mid H_a] \] 
$P_U$ is the Type II error.
\end{enumerate}

Like the \Bravo audit, this paper focuses on tests with $P_U = 0$. The risk, on the other hand, is an important (generally) non-zero value characterizing the quality of the audit. 

\subsection{Related Work}
\label{sec:rla}
\noindent
A {\em risk-limiting audit (RLA)} with {\em risk limit} $\alpha$---as described by, for example, Lindeman and Stark \cite{RLA}---is one for which the risk is smaller than $\alpha$ for all possible (unknown) true tallies in the election. For convenience when we compare audits, we refer to this audit as an $\alpha$-{\em RLA}.

\begin{definition}[Risk Limiting Audit ($\alpha$-RLA)]\label{def:rla}
 An audit $\mathcal{A}$ is a \textit{Risk Limiting Audit} with risk limit $\alpha$ iff for sample $X$
 \[
P[\mathcal{A}(X) = \textit{Correct} \mid H_0] \leq \alpha
 \]
\end{definition}

There are many audits that would satisfy the $\alpha$-{\em RLA} criterion, and not all would be desirable. For example, the constant audit which always outputs {\em Incorrect} always requires a hand count and is risk-limiting with $P_M=0 < \alpha$, $\forall \alpha$, $\forall p$. However, $P_U=1$, and the audit examines all votes each time; this is undesirable.

An example of an $\alpha$-{\em RLA} with $P_U=0$ and drawing fewer ballots is the \B~\Bravo audit \cite{bravo} which specifies round size increments of one. 

\begin{definition}[\Bravo]\label{def:bravo}  An audit $\mathcal{A}$ is the \B~$(\alpha, p)$-\Bravo audit iff the following stopping condition is tested at each ballot draw. If the sample $X$ is of size $n$ and has $k$ ballots for the winner,  
\begin{equation}
\mathcal{A}(S) =  \left\{ \begin{array}{ll} Correct & ~\sigma(k, p, n) \triangleq \frac{p^{k} (1-p)^{n-k}}{(\frac{1}{2})^n} \geq \frac{1}{\alpha}\\
Undetermined & ~else 
\end{array}
\right .
\label{eqn:bravo}
\end{equation}
Its p-value is ${\sigma (k, p, n)}^{-1}$. 
\end{definition}

\noindent
$\sigma(k, p, n)$ is the {\em likelihood ratio} of the drawn sequence~$X$. 
\noindent
The \B $(\alpha, p)$-\Bravo audit is an {\em SPRT} \cite{wald} with:
\begin{description}
\item $H_0$, the null hypothesis: the election is a tie
\item $H_a$, the alternate hypothesis: the fractional tally for the winner is $p$. 
\end{description}

\noindent
Implicit in Definition \ref{def:bravo} is the point that a sequence $X$ is tested only if it has not previously satisfied the test. If $\mathcal{A}(X_*) ~=~ Correct$ for some sequence $X_*$, all extensions $X_*^+$ of $X_*$ are defined as having passed the test.  Determining the stopping condition by evaluating $\mathcal{A}(X_*^+)$ does not satisfy the assumptions of the test, and the properties of the test do not necessarily apply. As we shall see in Section \ref{sec:problem}, this is relevant to {\em end-of-round} \Bravo. In fact, it is relevant to {\em end-of-round} applications of any \B audit that is an {\em SPRT}. 

\B \Bravo is a most efficient test given the hypotheses (if the stopping condition is satisfied exactly every time). Vora shows \cite{Bayesian-RLA} that \B $(\alpha, p)$-\Bravo is an $\alpha$-{\em RLA} because it assumes a tie for $H_0$, which is the wrong election outcome that is hardest to distinguish from the announced one, and hence defines the worst-case risk \cite{Bayesian-RLA}.


Other approaches, such as Rivest's CLIP Audit \cite{clip}, improve on \B \Bravo's efficiency subject to certain constraints (namely, of $\beta$ as defined in \cite{clip}). 

A prototype of \Athena mirrored the explicit risk allocation found in Stark's Conservative Statistical Post-Election Audits~\cite{stark2008} before ballots are examined for the audit, a list of increasing rounds $(n_1, n_2, ..., n_j)$,  and a list of corresponding risks $(\alpha_1, \alpha_2, ..., \alpha_j)$ are generated. Dispensing with auditor flexibility in favor of a predetermined list of rounds and corresponding risks facilitated the investigation of the convolution procedure that underlies a fundamental improvement of \Athena over \Bravo.

There is a line of work on \textit{group sequential testing}~\cite{groupSequential,heard2018choosing,fisher,ghosh} but all results that we were able to find begin with the assumption of a normal distribution and cannot be directly applied to the considered scenario of auditing elections.

\section{The Problem}
\label{sec:problem}
In this section we use an example to illustrate the problems of using \B rules for an \R audit. 

The \B $(\alpha, p)$-\Bravo audit, Definition \ref{def:bravo}, is the following ratio test (inequality (\ref{eqn:bravo})) performed after each draw:

\[ \sigma(k, p, n) = \frac{p^k(1-p)^{n-k}}{(\frac{1}{2})^n} \geq \frac{1}{\alpha}\]

Because $p > 1-p$ and the denominator above does not depend on $k$, $\sigma(k, p, n)$ is monotone increasing with $k$. There is hence a minimum value of $k$ for which the \B $(\alpha, p)$-\Bravo stopping condition is satisfied. That is, $ \exists~k_{min}(\Bravo, n, p, \alpha)$ such that the stopping condition of Definition \ref{def:bravo}, inequality (\ref{eqn:bravo}), is: 
\[ \mathcal{A}(S) ~=~ \textit{Correct} \Leftrightarrow k \geq  k_{min}(\Bravo, n, p, \alpha) \]

In fact it is easy to see that $k_{min}(\Bravo, n, p, \alpha)$  is a discretized straight line as a function of $n$, with slope and intercept determined by $p$ and $\alpha$ (see, for example, \cite{wald}). 

\begin{equation}
\label{eqn:kmin}
k_{min}(\Bravo, n, p, \alpha) = \left\lceil{m(\Bravo, p, \alpha) \cdot n + c(\Bravo, p, \alpha)}\right\rceil 
\end{equation}
where 
\[ m(\Bravo, p, \alpha) = \frac{\log{\frac{\frac{1}{2}}{1-p}}}{\log{\frac{p}{1-p}}} \]
\[c(\Bravo, p, \alpha) =  - \frac{\log{\alpha}}{\log{\frac{p}{1-p}}} \]
We drop one or more arguments of $k_{min}$, $c$ or $m$ when they are obvious.

\begin{example}[B2 \Bravo vs R2 \Bravo]
\label{eg:B2-vs-R2}
Let $\alpha=0.1$ and $p=0.75$, we get, from equation (\ref{eqn:kmin}): \[ k_{min}(\Bravo, n, 0.75, 0.1) \approx \lceil{0.6309n + 2.0959}\rceil \]

Consider ballots drawn in rounds of size $20, 40, 60, \ldots$ and the \Bravo condition being tested:  
\begin{itemize}
\item {\em End-of-Round}, which requires a record simply of the tally of the sample polled. 
\item {\em Selection-ordered-ballots}, requires a record of the vote on each ballot polled, in selection order. 
\end{itemize}
Note that the stopping condition is always the \Bravo stopping condition; the variation is in when it is checked. 

Figure \ref{fig:example} is a plot of $k_{min}(\Bravo, n, 0.75, 0.1) $ as a function of round size. It also shows the results of the tests above, performed on an example sequence. 
\begin{itemize}
\item For a hypothetical sequence, {\em selection-ordered-ballots} \Bravo checks the stopping condition at the blue squares till the stopping condition is satisfied, and the audit stops. It has information about the number of ballots for the winner and the total number of ballots drawn at each ballot draw. 
\item If the same sequence were to go through an  {\em end-of-round} \Bravo audit, the stopping condition would be checked only at the end of the round, denoted in the figure by black crosses. The audit only has information on vote tallies at the end of the round. 
\end{itemize}

\begin{figure}[h]
\centering
\includegraphics[scale=0.25]{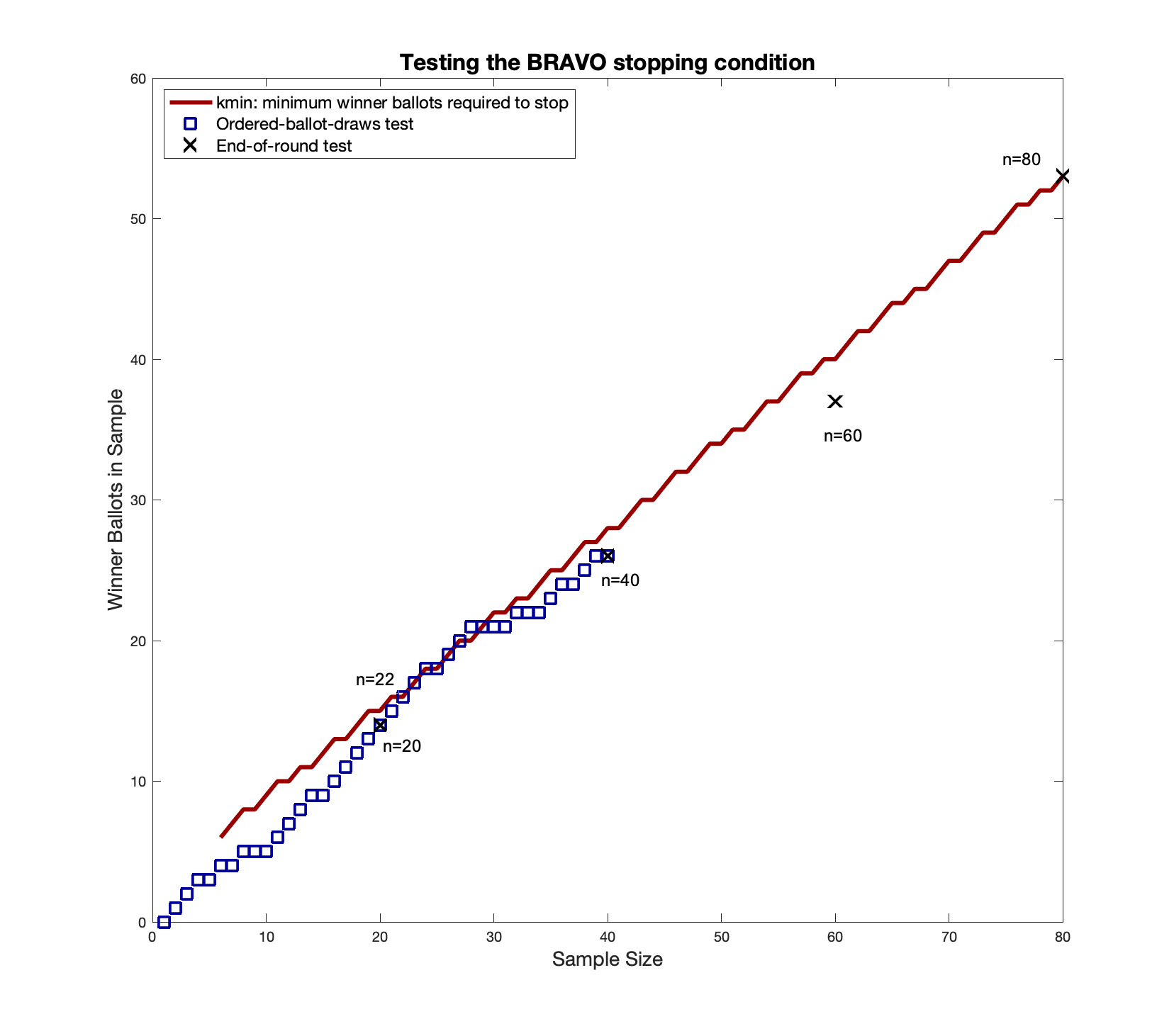}
\caption{Using \Bravo for a round-by-round audit with $p=0.75$, $\alpha=0.1$ and round size = $20$.}
\label{fig:example}
\end{figure}

We see that the stopping condition is satisfied during the second round, at $n=22$, but that it is no longer satisfied when it is tested at the end of that round, at $n=40$, or the following round, $n=60$. It is satisfied at the end of the fourth round, $n=80$, which is the number of ballots drawn in an {\em end-of-round} \Bravo audit. 

Thus: 
\begin{itemize}
\item \B \Bravo ends at $n=22$, and $22$ ballots are drawn. 
\item {\em End-of-round} \Bravo ends at $n=80$ and $80$ ballots are drawn.
\item {\em Selection-ordered-ballots} \Bravo ends at $n=22$, and $40$ ballots are drawn. 
\end{itemize}
\end{example}

The instance of {\em selection-ordered-ballots} \Bravo in our example would stop at the end of the second round after $40$ ballots are drawn. Such an audit is risk-limiting even though the condition is not satisfied at the $40^{th}$ draw. This is because every time a sequence $X$ satisfies the stopping condition, all extensions of it are defined as having passed the audit as well. In the event that the election outcome is incorrect, any sequence that passes the audit contributes to the risk. A risk-limiting audit ensures that the total risk contribution of all sequences that satisfy the audit is bounded above by the risk limit, whatever the underlying election. This accounting naturally includes risk contributions of all extensions of sequences that pass the audit as well. 
 
Note, however, that {\em selection-ordered-ballots} discards the extra information contained in the $18$ ballots drawn following the $22$-ballot draw. It ought to be possible to include this information, obtained at some cost, to better estimate the correctness of the election outcome. (Imagine telling election officials and the public that the p-value of the draw was small enough earlier, that it is not any more, and the math allows us to use the earlier value because if the election outcome is incorrect, it is accounted for in the risk limit). We need not be limited by the \B \Bravo rules which begin with a large disadvantage when used for \R audits, as they do not take into account that the ballots are drawn in rounds. 

\section{An Introduction to the \Athena Class of Audits}
\label{sec:informal}
In this section, we use an example to illustrate the workings of a proposed new \R audit \Minerva. In later sections, we provide more rigorous descriptions of the \Athena class of \R audits which we prove are risk-limiting and at least as efficient as {\em end-of-round} \Bravo. As we mentioned in section \ref{sec:intro}, our proof requires that the the round schedule be pre-determined (before the audit begins). For example, one may choose a factor $a$ such that $n_{j+1}-n_j = a(n_j-n_{j-1})$. 

\begin{example}[{\em End-of-Round} $(0.1, 0.75)$-\Bravo]\label{sec:informal_example}
We consider the {\em end-of-round} $(0.1, 0.75)$-\Bravo audit as in the previous section. Denote by $n_1$ the number of ballots drawn in the first round, and by $k_1$ those for the winner. 

Suppose $n_1=50$. Figure \ref{fig:tails} shows the probability distributions of $k_1$ for the two hypotheses:
\begin{description}
\item $H_a$: the election is as announced, with $p = 0.75$ (blue solid curve), and
\item $H_0$: the election is a tie (red dashed curve).
\end{description}

We will continue to refer to Figure \ref{fig:tails} in the following sections, when we will address the shaded areas. 

\begin{figure}[h]
\begin{centering}
 \includegraphics[scale=0.3]{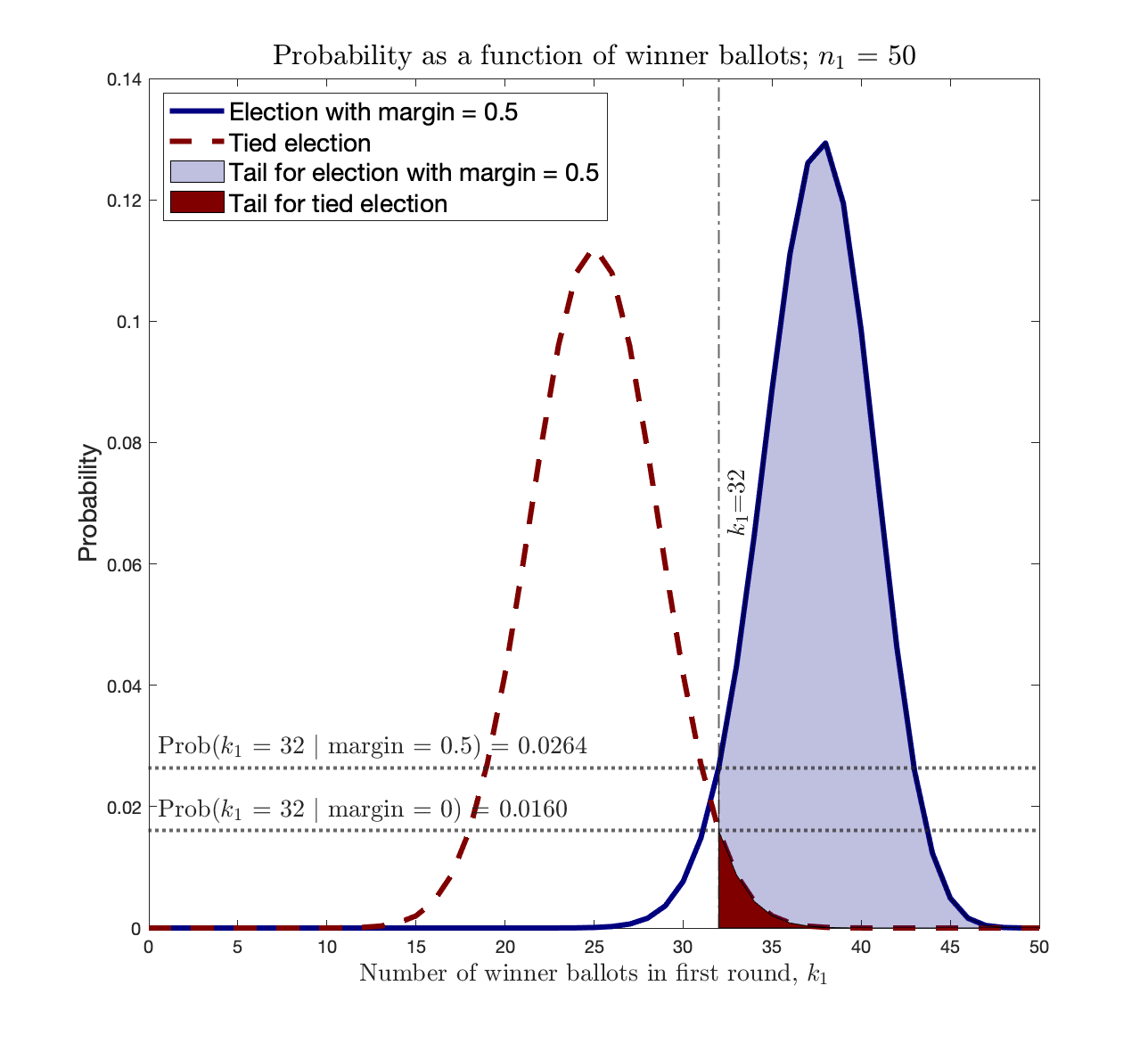}
 \caption{Probability Distribution of Winner Votes for $p=0.75$ and $n_1=50$: First Round. }
 \label{fig:tails}
 \end{centering}
\end{figure}

Suppose $k_1=32$. The \B $(0.1, 0.75)$-\Bravo stopping condition (see inequality(\ref{eqn:bravo})) tested {\em end-of-round} is: 
\begin{equation}
\label{eqn:bravo_prob_ratio2} 
\sigma(k_1, p, n_1) = \frac{p^{k_1} (1-p)^{n_1-k_1}}{(\frac{1}{2})^{n_1}} = \frac{{n_1 \choose k_1} p^{k_1} (1-p)^{n_1-k_1}}{{n_1 \choose k_1} (\frac{1}{2})^{n_1}} = \frac{Pr[k_1=32 \mid H_a] }{Pr[k_1=32 \mid H_0] }  \geq \frac{1}{\alpha}
\end{equation}

For our particular example, Figure \ref{fig:tails}, the likelihood ratio above is: 
\[ \sigma(32, 0.75, 50) = \frac{Pr[K_1=32 \mid H_a] }{Pr[K_1=32 \mid H_0] }  \approx \frac{0.0264}{0.0160} \approx  1.65 \not \geq \frac{1}{\alpha} = 10\]

And the sample does not pass the {\em end-of-round} \Bravo audit. Recall that the \B \Bravo p-value is the reciprocal of the above probability ratio. In this example, it is $\approx 0.6061 > \alpha = 0.1$. 

This is consistent with the fact that (see Example \ref{eg:B2-vs-R2}, Section~\ref{sec:problem}): 

\begin{equation}
\label{eqn:ex:bravo_kmin}
32 < k_{min}(\Bravo, 50, 0.75, 0.1) = \left\lceil{0.6309 \cdot 50 + 2.0959}\right\rceil = 34
\end{equation}
\end{example}

\subsection{The \Minerva Audit}

We propose the \Athena class of audits, which use the tails of the probability distribution functions to define the stopping condition. Here we provide an informal description of the simplest of the \Athena class, the \Minerva audit. 

\begin{example}[The \Minerva Audit]
\label{ex:minerva}
For the parameters of Example \ref{sec:informal_example}, $\alpha=0.1$, $p=0.75$, $n_1=50$ and $k_1=32$, we describe the \Minerva stopping condition, a comparison test of the ratio of the tails of the distributions: 
\begin{equation}
\label{eqn:minerva_prob_ratio}
\tau _1 (32, p, n_1) = \frac{Pr[K_1 \geq 32 \mid H_a, n_1] }{Pr[K_1 \geq 32 \mid H_0, n_1] } \geq \frac{1}{\alpha} 
\end{equation}
Compare this to the stopping condition for \Bravo, inequality (\ref{eqn:bravo_prob_ratio2}). 

Note that  $Pr[K_1 \geq 32 \mid H_a] $ is the stopping probability for round $1$ (the probability that the audit will stop in round $1$ given $H_a$) associated with deciding to stop at $k_1=32$---and not at smaller values. It is the tail of the solid blue curve, the translucent blue area in Figure \ref{fig:tails}. Similarly, $Pr[K_1=32 \mid H_0]$ is the associated risk. It is the tail of the red dashed curve denoting the tied election, and shaded red. 
 
For our example, the ratio of the tails of the two curves of Figure \ref{fig:tails} is (the values are not denoted in the figure): 
\[ \tau_1 (32, 0.75, 50) = \frac{Pr[K_1 \geq 32 \mid H_a, n_1] }{Pr[K_1 \geq 32 \mid H_0, n_1] }   \approx \frac{0.9713}{0.0325} \approx 29.89 > \frac{1}{\alpha} = 10 \]

And the sample passes the \Minerva audit. 
\end{example}

\subsection{The \Minerva audit is risk-limiting}
We will prove in Section~\ref{sec:theorems} that the \Minerva stopping condition is monotonic increasing as a function of $k_1$; one may understand this informally as follows. As explained in Section~\ref{sec:problem}, $\sigma(k_1, p, n_1)$ is monotone increasing with $k_1$. The \Minerva ratio $\tau_1(k_1, p, n_1)$ is a weighted average of the values of $\sigma(k, p, n_1)$ for $k \geq k_1$ and is also, hence, monotone increasing with $k_1$. If a sample with $32$ winner ballots of a total of $50$ ballots were to satisfy the stopping condition, so would all samples with $k_1 \geq 32$. 

Smaller values of $k_1$ are associated with larger tails in both curves of Figure \ref{fig:tails}; and the tails denote the stopping probability (given $H_a$, the translucent blue tail of the solid blue curve) and the risk (given $H_0$, the solid red tail of the dashed red curve). The smaller the value of $k_1$, hence, the larger the associated stopping probability and risk. We could simply choose the smallest acceptable $k_1$ (denoted as $k_{min, 1}$) so that the associated risk is $\alpha$, but then we could not plan to ever go to another round because we would have exhausted the risk budget in the first round. For the Minerva audit, we choose $k_1$ so that the risk is no larger than $\alpha$ times the stopping probability. This allows us to go on to an indefinite number of rounds. 

Let $R_j$ and $S_j$ be informally defined as follows (more formal definitions follow in Sections~\ref{sec:algorithms}~and~\ref{sec:computing}): 
\[ R_j = Pr[\Minerva \text{~audit~stops~in~round~$j$~and~no~earlier} \mid H_0] \] 
and 
\[ S_j = Pr[\Minerva \text{~audit~stops~in~round~$j$~and~no~earlier} \mid H_a] \] 
We define $R_j$ and $S_j$ more carefully in section \ref{sec:theorems} and describe how to compute these values in section \ref{sec:computing}.  Loosely speaking, they denote the risk associated with the $j^{th}$ round ($R_j$) and the stopping probability  of the $j^{th}$ round ($S_j$) respectively. 

The \Minerva stopping condition is: 
\[ \frac{R_j}{S_j} \leq \alpha \Rightarrow R_j \leq \alpha \cdot S_j \]
If $R$ is the risk of the audit and $S$ its stopping probability,  
\[ R = \sum _j R_j \leq \alpha \sum _j S_j \leq \alpha \cdot S \leq \alpha \]
because $S$, the stopping probability of the audit, is no larger than $1$. 

In other words, the total risk of the audit is the sum of the risks of each individual round. The stopping condition ensures that each of these risks is no larger than $\alpha$ times the corresponding stopping probability. Adding all the risks gives us the total risk, which is no larger than $\alpha$ times the total stopping probability. Because the total stopping probability cannot be larger than one, the total risk cannot be larger than $\alpha$, and \Minerva is risk-limiting.

\subsection{\Minerva is at least as efficient as {\em end-of-round} \Bravo}
In this section, we further examine the audit of our previous examples to understand the behavior of the ratios of \Bravo and \Minerva, $\sigma(k_1, p, n_1)$ and $\tau_1(k_1, p, n_1)$ respectively. 
\begin{example}[\Bravo vs. \Minerva Ratios]\label{sec:make-up}
For the parameters of Examples \ref{sec:informal_example} and \ref{ex:minerva}: $p=0.75$, $\alpha=0.1$ and $n_1=50$, Figure \ref{fig:ratios_both} presents the likelihood ratio for {\em end-of-round} \Bravo (green solid line), $\sigma(k_1, 0.75, 50)$, and the tail ratio for \Minerva (orange dashed line), $\tau _1(k_1, 0.75, 50)$, on a log scale. An audit satisfies the stopping condition when its ratio equals or exceeds $\alpha ^{-1} = 10$. 
\begin{figure}[h!]
\begin{centering}
 \includegraphics[scale=0.35]{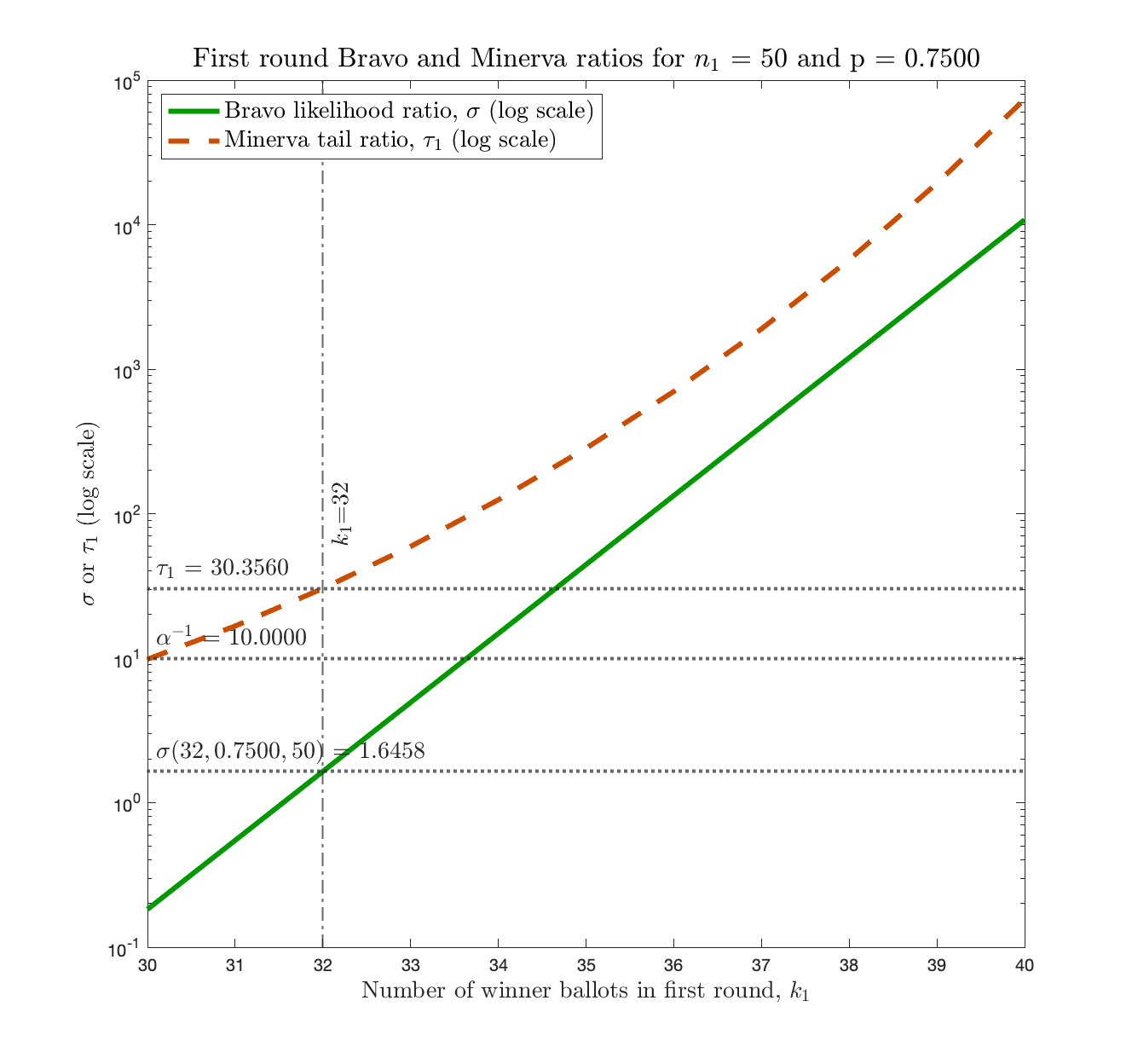}
 \caption{BRAVO and Minerva comparison tests for $p=0.75$ and $n_1=50$: First Round. }
 \label{fig:ratios_both}
\end{centering}
\end{figure}

We see that 
\[ \sigma(k_1, 0.75, 50) < \tau _1(k_1, 0.75, 50) \]

This means that any sample satisfying {\em end-of-round} \Bravo will also satisfy \Minerva. In fact, it will often be the case that the \Minerva condition will be satisfied and the {\em end-of-round} \Bravo one will not. 

We have seen earlier (see equation (\ref{eqn:ex:bravo_kmin}), Example \ref{sec:informal_example}) that 
\[ k_{min}(\Bravo, 50, 0.75, 0.1) = 34\]
and {\em end-of-round} \Bravo will stop for $k_1 \geq 34$ and no smaller values of $k_1$. On the other hand, we see from Figure \ref{fig:ratios_both} that \Minerva would stop additionally for $k_1=31, 32, 33$. 
\end{example}
The reason for \Minerva stopping at smaller values of $k_1$ is as follows. Consider $k_1=32$. While 
\[ \sigma(32, 0.75, 50) \approx 1.6458 < 10 \]
$\sigma(k_1, 0.75, 50)$ can be much larger for larger values of $k_1$. The \Minerva ratio, $\tau _1$, at $k_1=32$ is a weighted average of all the values of $\sigma(k_1, 0.75, 50)$ for $k_1 \geq 32$, allowing the larger values of $\sigma(k_1, 0.75, 50)$ to ``make up'' for the smaller ones; in fact, $\tau _1 = 30.356$ for $k_1=32$. 

In other words, because the {\em end-of-round} \Bravo ratio increases as $k_1$ increases, the weighted average, the \Minerva ratio, will always be larger than the {\em end-of-round} \Bravo ratio except if $k$ is the largest possible number of winner votes, in which case the two ratios will be equal.  Equivalently, the \Minerva p-value will always be smaller except when $k$ is the largest possible number of winner votes, and the p-values are equal. Thus, \Minerva is at least as efficient. 

\subsection{The \Athena audit}
In this section we present an example introducing the \Athena audit. 
\begin{example}
We can see from Figure \ref{fig:ratios_both} that the \Minerva tail ratio at $k=31$ is larger than $\alpha ^{-1}$. However, the {\em end-of-round} \Bravo ratio is smaller than $1$: \[ \sigma(31, 0.75, 50) < 1 \Rightarrow Pr[K_1 = 31 \mid H_a] < Pr[K_1 = 31 \mid H_0] \]
which means that the observation $K_1=31$ is more likely given $H_0$ (the election is a tie) than it is given $H_a$ (the election is as announced)! 

This is technically not an issue; it simply means that \Minerva can stop in such a situation and be risk-limiting. We could also choose to enforce an additional stopping condition of a lower bound on the likelihood ratio of the sample. If the sample $x$ satisfies the $\Minerva$ condition and is of size $n$ with $k^*$ votes for the winner, the additional stopping condition would be:  \[ \frac{Pr[K=k^* \mid H_a]}{Pr[K=k^* \mid H_0]} = \sigma(k^*, p, n) \geq \frac{1}{\delta}\]
for some $\delta$. We term this combination of two stopping conditions the \Athena audit. 

A reasonable choice is $\delta=1$ (the observation is at least as likely given $H_a$ as it is given $H_0$). We would, of course, not desire $\delta < \alpha$, because we would then be requiring the satisfaction of the \Bravo condition with risk limit $\delta < \alpha$. 

We can see from Figure \ref{fig:ratios_both} that the \Athena condition for $\delta=1$ is satisfied for $k^* \geq 32$ and no smaller values of $k_1$. Recall that \Minerva stops for $k_1 \geq 31$. Thus, \Minerva would stop for $k_1=31$ and \Athena would not. 
\end{example}

In our experiments we have observed samples satisfying \Minerva but not \Athena when the election margin is wide, as in our example. Hence, clearly, \Athena is not as efficient as \Minerva, because it imposes an additional condition. One may think of \Minerva as determining whether the election outcome is correct, and \Athena determining, in addition, if the election tally is close enough to the announced tally.

\section{Computing Risks and Stopping Probabilities for Multiple-Round Audits}
\label{sec:computing}
In this section we describe how probability distributions may be computed in multiple round audits with monotone stopping conditions; that is, audits where the stopping condition is represented through the use of $k_{min}$. We use examples to demonstrate how the probability distributions may be computed for rounds $2$ and above. 

\begin{example}[Testing the Stopping Condition]
\label{ex:lopping}
Consider an election with $p=0.75$ and a risk limit of $\alpha = 0.1$. Suppose the first round size is $n_1=50$ and the draw results in $k_1=30$ ballots for the announced winner. Recall that (see equation (\ref{eqn:ex:bravo_kmin}), Example \ref{sec:informal_example})
\[ k_{min}(\Bravo, 50, 0.75, 0.1) = 34. \]
and (see Figure \ref{fig:ratios_both}, Example \ref{sec:make-up})
\[k_{min}(\Minerva, 50, 0.75, 0.1) = 31\]

Thus the sample passes neither the \Minerva nor the {\em end-of-round} \Bravo audit. 

Now suppose we draw $50$ more ballots to get $n_2=100$ ballots in all, of which $k_2$ are for the winner. We will need to compute the probability distribution on $k_2$ to determine the ratio of the tails for the \Minerva stopping condition. 

Note that the probability distribution of $k_2$ is {\em not} the binomial distribution for a sample size of $100$. In fact, if the audit did not stop in the first round, $k_1 < 31$ for \Minerva, which means that $k_2$ can be no larger than $80$, even if all $50$ ballots in the second round are for the announced winner. Similarly, $k_2$ for the {\em end-of-round} \Bravo audit can be no larger than $83$. 

If the audit continues, the maximum number of ballots before new ones are drawn is $33$ for \Bravo and $30$ for \Minerva. The probability distributions before the new sample is drawn are as shown in Figures \ref{fig:bravo:lopped} and \ref{fig:minerva:lopped}, and may be denoted as: 
\[ Pr[k_1 \wedge (\mathcal{A_B}(X_1) \neq ~Correct) \mid H_a], ~~ Pr[ k_1 \wedge (\mathcal{A_B}(X_1) \neq ~Correct) \mid H_0] ~~ \]
and
\[ Pr[k_1 \wedge (\mathcal{A_M}(X_1) \neq ~Correct) \mid H_a], ~~ Pr[k_1 \wedge (\mathcal{A_M}(X_1) \neq ~Correct) \mid H_0]\]
where $\mathcal{A}_B$ and $\mathcal{A}_M$ denote the {\em end-of-round} \Bravo and \Minerva audits for the given parameters. 

\begin{figure}[h!]
\begin{centering}
 \includegraphics[scale=0.25]{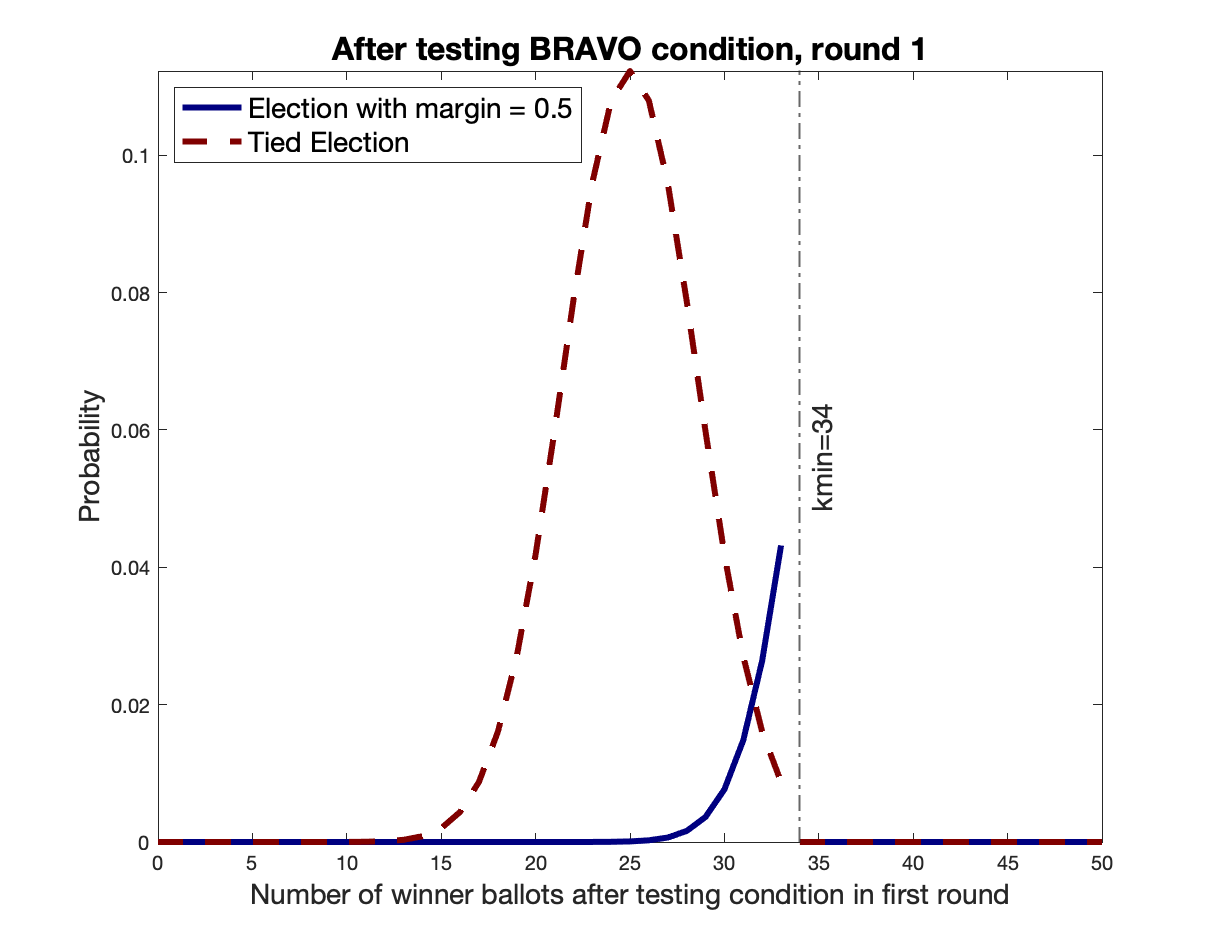}
 \caption{Probability Distribution of Winner Ballots for {\em end-of-round} \Bravo: $p=0.75$, $n_1=50$: After Testing the Stopping Condition for the First Round. }
 \label{fig:bravo:lopped}
 \end{centering}
\end{figure}

\begin{figure}[h!]
\begin{centering}
 \includegraphics[scale=0.25]{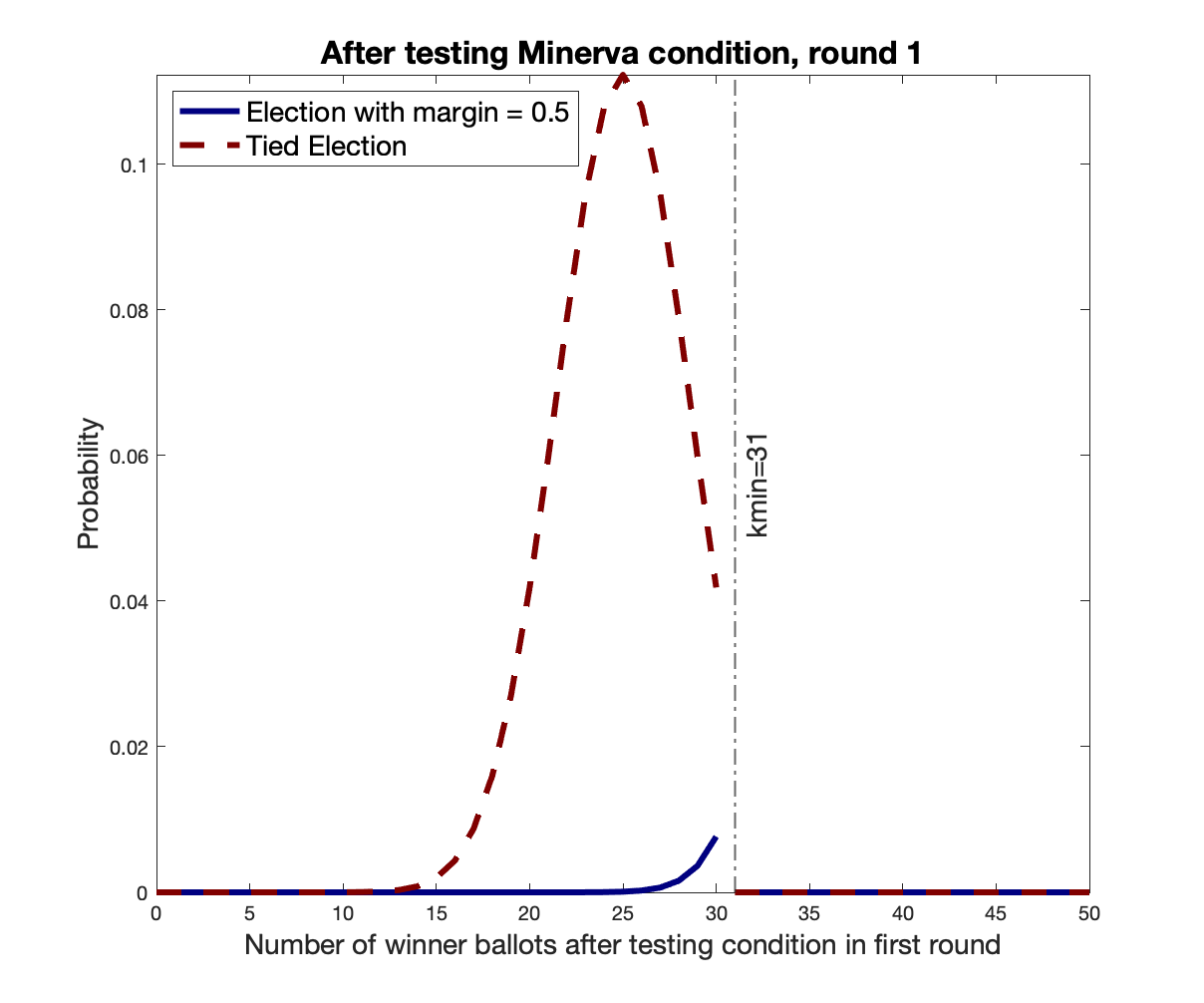}
 \caption{Probability Distribution of Winner Ballots for \Minerva: $p=0.75$, $n_1=50$: After Testing the Stopping Condition for the First Round. }
 \label{fig:minerva:lopped}
 \end{centering}
\end{figure}
\end{example}

The ``discarded'' tails, in both cases, represent the probabilities that the audit stops. When this is conditional on $H_a$, we refer to it as the stopping probability of the round ($S_1$), large values are good. When it is conditional on $H_0$, it is the worst-case risk corresponding to the round ($R_1$), large values are bad. Recall that our stopping condition bounds the worst-case risk for the round to be no larger than a fraction $\alpha$ of the stopping probability.

Using the above probability distributions, we can now compute the distribution of ballots for the announced winner in the sample of size $100$, which we obtain after drawing $50$ more ballots. 

\begin{example}[Second Round Distribution]
Continuing with Example \ref{ex:lopping}, we consider an election with $p=0.75$, risk limit $\alpha = 0.1$ and round sizes $n_1=50, n_2=100$. We wish to compute the probability distribution for $K_2$, the number of votes for the announced winner after drawing the second round of ballots. 
Recall that (see equation (\ref{eqn:ex:bravo_kmin}), Example \ref{sec:informal_example})
\[ k_{min}(\Bravo, 50, 0.75, 0.1) = 34. \]
and (see Figure \ref{fig:ratios_both}, Example \ref{sec:make-up})
\[k_{min}(\Minerva, 50, 0.75, 0.1) = 31\]

First consider the {\em end-of-round} \Bravo audit. After the first round stopping condition is tested, and the audit stopped if the condition is satisfied, the number of votes for the winner is at most $33$. Let $K_1$ be the number of votes for the winner. Then $K_1$ lies between $0$ and $33$. It is distributed as in Figure \ref{fig:bravo:lopped}, and we denote the distribution by $f(k_1 \mid H_0)$ for the null hypothesis (tied election, represented by the red dashed line) and $f(k_1 \mid H_a)$ for the alternate hypothesis (election is as announced, represented by the blue solid line). 

There would be a total of $K_2=k_2$ winner ballots in the sample after the second draw if $k_2-k_1$ winner ballots were drawn among the $50$ new ballots drawn in round 2. $k_2-k_1$ is a random variable, and its distribution is the binomial distribution for the draw of size $50$. 

If we denote the distribution of $K_2$ as $g$, it is:  
\[ g(k_2 \mid H) = \sum _{k_1=max\{0, k_2-50\}}^{min\{k_{min}-1,k_2\}} f(k_1 \mid H) \cdot binomial(k_2-k_1, 50, H) \]
where $binomial(j, n, H)$ is the probability of drawing $j$ votes for the announced winner in a sample of size $n$, when the fractional vote for the announced winner is $\frac{1}{2}$ for $H=H_0$ and $p$ for $H=H_a$. 

The above expression is known as the convolution of the two functions, and is denoted: 
\[ g(\cdot \mid H) = f(\cdot \mid H) \circledast   binomial(., 50, H) \]
where $\circledast$ represents the convolution operator and $H$ the hypothesis. The convolution of two functions can be computed efficiently using Fourier Transforms; this result is the convolution theorem. 

After drawing the second sample, the probability distributions for \Bravo are as in Figure \ref{fig:bravo:conv}. 
\begin{figure}[h!]
\begin{centering}
 \includegraphics[scale=0.25]{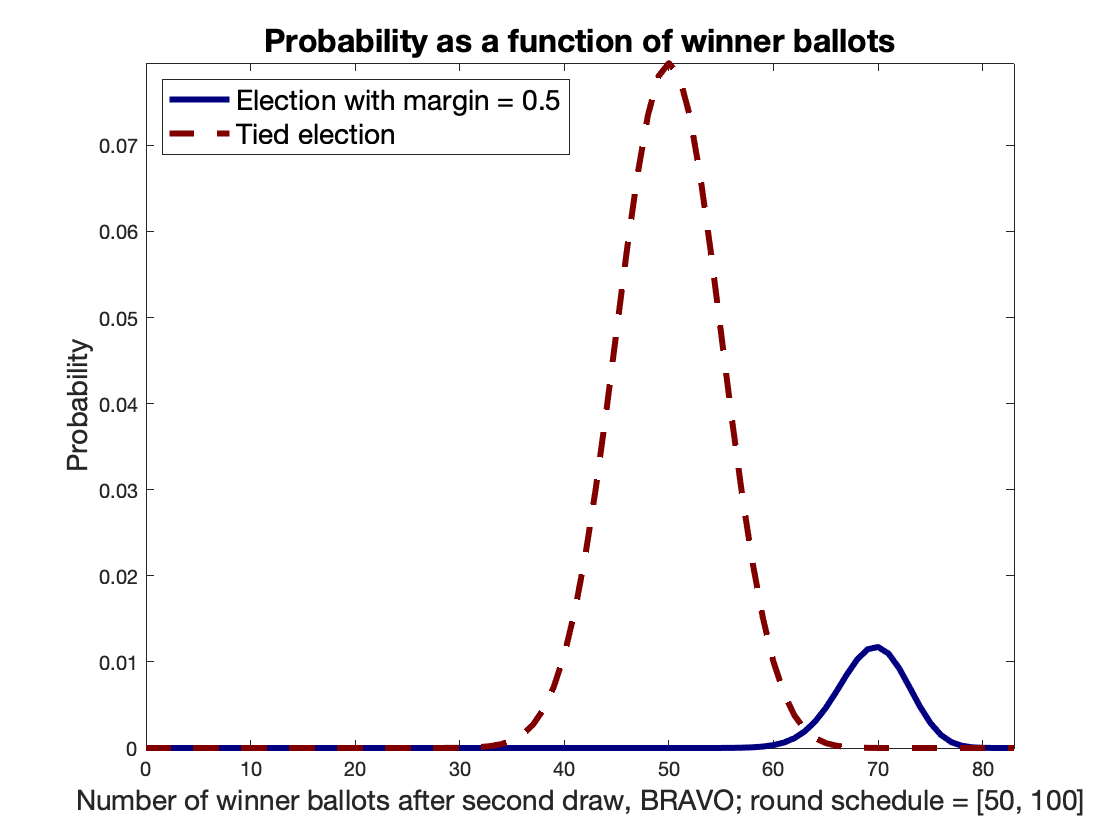}
 \caption{Probability Distribution of Winner Ballots for {\em end-of-round} \Bravo: $p=0.75$, $n_1=50$, $n_2=100$: After Drawing the Second Round. }
 \label{fig:bravo:conv}
 \end{centering}
\end{figure}

Similarly, for \Minerva, $k_1$ lies between $0$ and $30$. $k_2$ is similarly the convolution of the function(s) represented in Figure \ref{fig:minerva:lopped} and the binomial distribution corresponding to a draw of $50$ ballots for the respective hypotheses. After drawing the second sample, the probability distributions for \Minerva are as in Figure \ref{fig:minerva:conv}. 
\begin{figure}[h!]
\begin{centering}
 \includegraphics[scale=0.25]{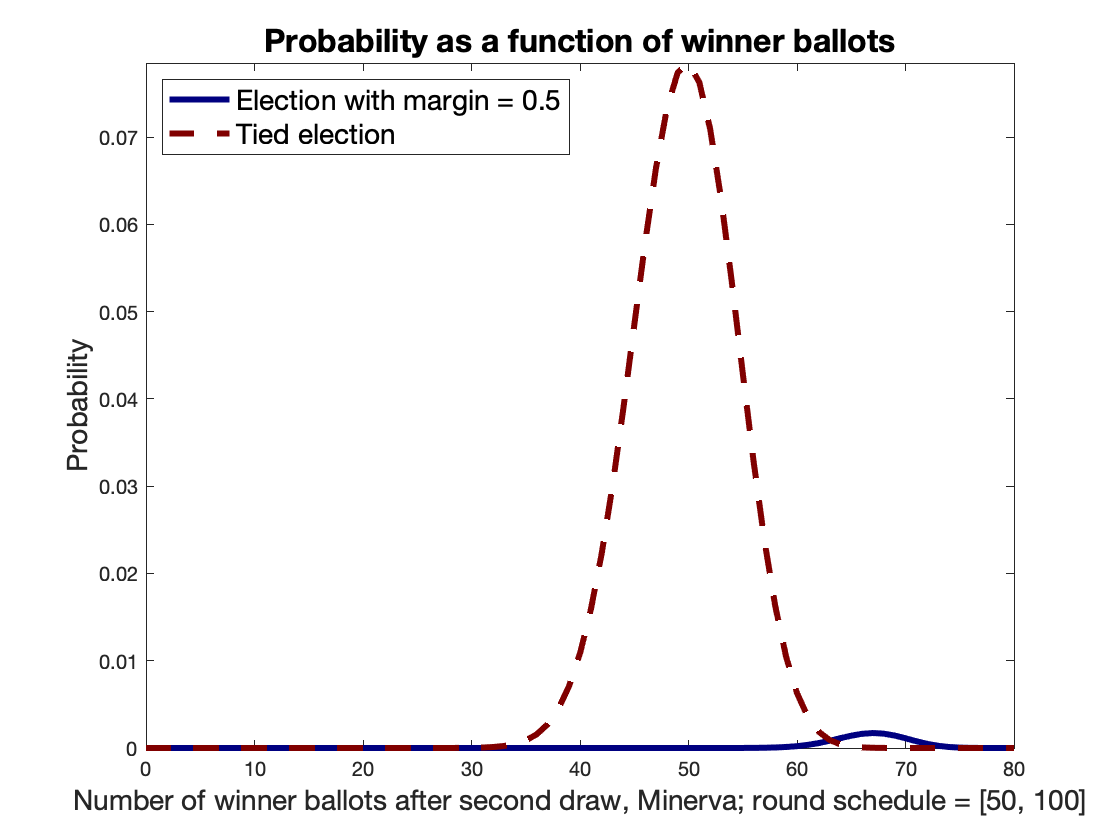}
 \caption{Probability Distribution of Winner Ballots for \Minerva: $p=0.75$, $n_1=50$, $n_2=100$: After Drawing the Second Round. }
 \label{fig:minerva:conv}
 \end{centering}
\end{figure}
\end{example}

In order to compute probability distributions for the next round, we would first compute the value of $k_{min, 1}$ for this round using the tail ratio, then zero the probability distributions for the value of $k_{min, 1}$ and above, and then perform a convolution with the binomial distribution corresponding to the size of the next draw. And so on. 

Probability distributions for \B audits may be computed similarly, with the round schedule: $(1, 2, \ldots, i, \ldots)$.  We used this approach to compute percentiles for the \Bravo stopping probabilities; see Section \ref{sec:results} for the results. 

\section{The \Athena Class of Audits}
\label{sec:algorithms}
In this section we rigorously describe \Minerva and \Athena, two audits from the new \Athena class of risk-limiting audits. The stopping condition for \Bravo is a comparison test for the ratio of probabilities of the number of winner ballots. On the other hand, the stopping conditions for the \Athena class are comparison tests for the ratio of the complementary cumulative distribution functions (cdfs). For the \Athena class of audits, the stopping condition for a given round does depend on previous round sizes, which are required to compute the complementary cdfs, but not on future round sizes. 

 \subsection{The \Minerva audit}
Given the \B $(\alpha, p)$-\Bravo test we define the corresponding \R \Minerva test by its stopping condition, which is a comparison test of the ratio of the complementary cdfs of samples that did not satisfy the stopping condition for any previous round. We expect that similar \R \Minerva tests can be defined for other {\em SPRT}s with zero error of the second kind whose probability ratio is a monotonic increasing function of $k$, such as Bayesian audits, but do not address these in this paper. Note that the round schedule is predetermined before the audit begins. 

\begin{definition}[$ (\alpha, p, (n_1, n_2, \ldots, n_j, \ldots) ) $-\Minerva]
 \label{def:minerva}
Given \B $(\alpha, p)$-\Bravo and round sizes $n_1, n_2, \ldots, n_j, \ldots$, the corresponding \R \Minerva stopping rule for the $(j+1)^{th}$ round is:
\begin{equation}
\mathcal{A}(X_{j+1})=  \left\{ \begin{array}{ll} Correct & ~~~\tau_{j+1}(k_{j+1}, p, (n_1, n_2, \ldots, n_j, n_{j+1}), \alpha ) \geq \frac{1}{\alpha}\\
& \\
Undetermined~(draw~more~samples) & ~~~else \\
\end{array}
\right .
\label{eqn:minerva-test}
\end{equation}
where $\tau _{j+1}$ is the {\em complementary cumulative distribution ratio} for the ${j+1}^{th}$ round:
\begin{equation}
\tau_{j+1}(k_{j+1}, p, (n_1, n_2, \ldots, n_j, n_{j+1}), \alpha )  =  \frac{Pr[K_{j+1} \geq k_{j+1} \wedge \forall_{i\leq j} ({\mathcal{A}}(X_i) ~\neq Correct) \mid H_a, n_{j+1}]}{Pr[K_{j+1} \geq k_{j+1} \wedge \forall_{i\leq j} ({\mathcal{A}}(X_i) ~\neq Correct) \mid tie, n_{j+1}]} ~~~j \geq 1
\label{eqn:minerva_sigma}
\end{equation}
and, as with \B $(\alpha, p)$-\Bravo, $H_a$, the alternate hypothesis, is that the fractional tally for the winner is $p$. 
\end{definition}
Clearly, for $j=0$ and the first round, 
\[\tau _{1}(k_1, p, n_1) =  \frac{Pr[K_1 \geq k_{1} \mid H_a, n_1]}{Pr[K_1 \geq k_{1} \mid tie, n_1]}\]

\subsection{The \Athena audit}
In addition to comparing the ratio of complementary cumulative distribution functions as in \Minerva, the \Athena audit also enforces a lower bound on the {\em probability ratio}, $\sigma$, of the \B $(\alpha, p)$-\Bravo test. 

\begin{definition}[$ (\alpha, \delta, p, (n_1, n_2, \ldots, n_j, \ldots) ) $-\Athena]
\label{def:athena}
Given \B $(\alpha, p)$-\Bravo, round sizes $n_1, n_2, \ldots, n_j, \ldots$ and parameter $\delta$, the corresponding \Athena stopping rule for the $(j+1)^{th}$ round is:
\begin{equation}
\mathcal{A}(X_{j+1})=  \left\{ \begin{array}{ll} Correct & ~~~\omega_{j+1}(k_{j+1}, p, (n_1, n_2, \ldots, n_j, n_{j+1}) )  \geq \frac{1}{\alpha}  \\
& \wedge \sigma(k_{j+1}, p, n_{j+1}) \geq \frac{1}{\delta} \\
& \\
Undetermined~(draw~more~samples) & ~~~else \\
\end{array}
\right .
\label{eqn:athena-test}
\end{equation}
where $\omega _{j+1}$ is the {\em complementary cumulative distribution ratio} for the $(j+1)^{th}$ round:
\begin{equation}
\omega_{j+1}(k_{j+1}, p, (n_1, n_2, \ldots, n_j, n_{j+1}), \alpha, \delta ) =  \frac{Pr[K_{j+1} \geq k_{j+1} \wedge \forall_{i\leq j} ({\mathcal{A}}(S_i) ~\neq Correct) \mid H_a, n_{j+1}]}{Pr[K_{j+1} \geq k_{j+1} \wedge \forall_{i\leq j} ({\mathcal{A}}(X_i) ~\neq Correct) \mid tie, n_{j+1}]} ~~~ j \geq 1
\label{eqn:athena-ratio}
\end{equation}

\[\sigma(k_{j+1}, p, n_{j+1}) =  \frac{p^{k_{j+1}}(1-p)^{n_{j+1}-{k_{j+1}}}}{(\frac{1}{2})^{n_{j+1}}} ~~~ j \geq 1\]

and, as with \B $(\alpha, p)$-\Bravo, $H_a$, the alternate hypothesis, is that the fractional tally for the winner is $p$. 
\end{definition}

Clearly, for $j=0$ and the first round, 
\[\omega _{1}(k_1, p, n_1) =  \frac{Pr[K_1 \geq k_{1} \mid H_a, n_1]}{Pr[K_1 \geq k_{1} \mid tie, n_1]}\]

We further define the risk and stopping probability associated with each round. 
\begin{definition}[$S_j$]
\label{def:stopping_j}
The probability of stopping in the $j^{th}$ round for audit ${\mathcal A}$ is defined as: 
\[ S_j = Pr[({\mathcal{A}}(X_j) ~= Correct)  \wedge \forall_{i < j} ({\mathcal{A}}(X_i) ~\neq Correct) \mid H_a, n_{j}] \]
\end{definition}

\begin{definition}[$R_j$]
\label{def:risk_j}
The risk of the $j^{th}$ round of audit ${\mathcal A}$ is defined as: 
\[ R_j = Pr[({\mathcal{A}}(X_j) ~= Correct)  \wedge \forall_{i < j} ({\mathcal{A}}(X_i) ~\neq Correct) \mid H_0, n_{j}] \]
\end{definition}  

\section{Risk-Limiting Properties of the \Athena Class of Audits}
\label{sec:theorems}
In this section we present the risk-limiting and efficiency properties of \Minerva and \Athena. We begin with an outline of our approach. Rigorous statements follow and proofs are in the Appendix. 

\subsection{An outline of the proofs}
In this section we provide an outline of the claims and proofs. 

Using induction on the number of rounds, we prove a couple of interesting properties, including, at the core, that the likelihood ratio of $K_j$ (total number of winner ballots) in the $j^{th}$ round is 
\[ \frac{Pr[K_j=k \wedge \forall_{i < j} ({\mathcal{A}}(X_i) ~\neq Correct) \mid H_a, n_{j}]}{Pr[K_j = k \wedge \forall_{i < j} ({\mathcal{A}}(X_i) ~\neq Correct) \mid H_0, n_{j}]}  = \sigma(k, p, n) \triangleq \frac{p^k(1-p)^{n-k}}{(\frac{1}{2})^{n}} \]
That is, for $k$ winner ballots in round $j$, when sequences are restricted to those that did not satisfy stopping conditions in previous rounds, the likelihood ratio is simply $\sigma(k, p, n)$, independent of any additional constraints on sequence order and the past or future round schedule. This property leads to the result that the test in the $j^{th}$ round is a comparison test for $k_j$. 

For the base case, it is easily shown that the likelihood ratio in the first round is $\sigma(k, p, n)$, as there is no previous round:
\[ \frac{Pr[K_1=k \mid H_a, n]}{Pr[K_1=k \mid H_0, n]} = \frac{{n \choose k} p^k(1-p)^{n-k}}{{n \choose k} (\frac{1}{2})^{n}} = \sigma(k, p, n) \]
$\sigma(k, p, n)$ is easily seen to be monotone increasing with $k$ because $p > \frac{1}{2}$. 

The induction step proceeds as follows. Suppose the likelihood ratio for the number of winner ballots in round $j$ is $\sigma(k, p, n)$. The ratios used for the stopping conditions of the $j^{th}$ round in \Minerva and \Athena
\[ \tau_{j}(k, p, (n_1, n_2, \ldots, n_j), \alpha )\] and 
\[ \omega_{j}(k, p, (n_1, n_2, \ldots, n_j), \alpha,\delta )\]
respectively, are weighted sums of $\sigma(k_j, p, n)$ for $k_j \geq k$. Because $\sigma(k_j, p, n)$ is monotone increasing with $k_j$, the respective stopping conditions are monotone increasing with $k$, and can be expressed as comparison tests for $k$. 

\Athena has two conditions. The second one is a comparison test for the likelihood ratio, and hence also equivalent to a comparison test for the number of winner ballots. The overall comparison test is the stricter of the two, and is also a comparison test. Thus the stopping condition for round $j$ is a test of the form $k \geq k_{min, j}$ for a $k_{min, j}$ that depends on the audit, its parameters including previous round sizes, the election parameters and the risk limit. As described in section \ref{sec:informal}, one can use convolution to compute the probability distributions for $k_{j+1}$. 

To determine the nature of the likelihood ratio for round $j+1$, we proceed as follows. The likelihood ratio for $k$ in round $j$ is assumed to be $\sigma(k, p, n)$. Hence, the distribution of $k$ in round $j$ given $H_a$ is a multiple of $p^k(1-p)^{n-k}$, and that given $H_0$ is the same multiple of $(\frac{1}{2})^n$. The multiplying factor itself is a function of $k$, current and previous round sizes, $p$ and $\alpha$, as it captures the previous comparison tests on the number of winner ballots. On convolution, when one obtains the distributions for round $j+1$, the multiplying factors change, but the one for the distribution given $H_a$ is the same as that for the distribution given $H_0$. Loosely speaking, the number of ways of obtaining a sequence with $k$ winner ballots in a round of size $n$, given previous round sizes (and hence previous comparison tests for winner ballots), is independent of the hypothesis. Thus the likelihood ratio of $k$ in round $j+1$ remains $\sigma(k, p, n)$. 

It is now easy to show that the audits are risk-limiting. The tail (beginning at $k_{min, j}$) of the probability distribution of $K_j$ in the $j^{th}$ round conditional on hypothesis $H_a$ ($H_0$) is the stopping probability (risk) associated with the $j^{th}$ round. Both \Minerva and \Athena audits require that the tail corresponding to $H_0$ (the risk corresponding to round $j$) be no more than $\alpha$ times the tail corresponding to $H_a$ (the stopping probability), thus ensuring that the sum of all the risk contributions of the rounds is no more than $\alpha$ times the total stopping probability, and hence no more than $\alpha$. 

Both $\Minerva$ and $\Athena$ are more efficient than {\em end-of-round} \Bravo. This is because the ratios $\sigma$, $\tau_j$ and $\omega_j$ are all compared to the same value, $\frac{1}{\alpha}$. For $k$ winner ballots, the ratio for \Bravo is $\sigma(k, p, n)$, which is monotone increasing with $k$. The ratios for \Athena and \Minerva, $\omega_j$ and $\tau_j$ respectively, are weighted sums of $\sigma(k_j, p, n)$ for $k_j \geq k$ and hence strictly larger, unless $k$ is the largest value with non-zero probability (it would be value of of $k_{min}-1$ for the previous round, plus the size of the current draw), when they are equal. Thus, if a sample satisfies the {\em end-of-round} \Bravo condition, it also satisfies the conditions on $\omega_j$ and $\tau_j$. The \Athena audit includes a second test, which is a comparison test for $\sigma(k, p, n)$. If $\delta \geq \alpha$, this too is satisfied if the {\em end-of-round} \Bravo condition is satisfied. 

We state the above claims more formally in the rest of this section, and prove them in the Appendix. We also prove that the \Minerva and \Athena ($\delta \geq \alpha$) stopping conditions, when round increments are specified to be one, are equivalent to the \B \Bravo condition, though p-values are not the same except for samples where the audits stop. 

\subsection{Notation}

We establish some shorthand notation which will be useful. 

For ease of notation, when the audit and its parameters: round schedule $(n_1, n_2, \ldots)$, risk limit $\alpha$, fractional vote for the winner $w_a$ are fixed, we denote: 
\[S_{j} (k_{j}) \triangleq Pr[K_j \geq k_{j} \wedge \forall_{i < j} ({\mathcal{A}}(X_i) ~\neq Correct) \mid H_a, n_1, \ldots, n_{j}] \]
\[R_{j} (k_{j}) \triangleq Pr[K_j \geq k_{j} \wedge \forall_{i < j} ({\mathcal{A}}(X_i) ~\neq Correct) \mid H_0, n_1, \ldots, n_{j}] \]
Thus $\frac{S_j(k_j)}{R_j(k_j)}$ is the ratio of the complementary cdfs in round $j$ when the number of winner ballots drawn is $k_j$, and the sequence did not satisfy the stopping condition in a previous round. 

Similarly, 
\[s_{j} ( k_{j}) \triangleq Pr[K_j = k_{j} \wedge \forall_{i < j} ({\mathcal{A}}(X_i) ~\neq Correct) \mid H_a, n_1, \ldots, n_{j}] \]
and 
\[r_{j} (k_{j}) \triangleq Pr[K_j = k_{j} \wedge \forall_{i < j} ({\mathcal{A}}(X_i) ~\neq Correct) \mid H_0, n_1, \ldots, n_{j}] \]
and $\frac{s_j(k_j)}{r_j(k_j)}$ is the likelihood ratio of $k_j$ winner ballots in round $j$ when the sequence did not satisfy the stopping condition in a previous round. 

Note also the following simple observation: 
\begin{equation}
\begin{split}
S_{j} (k_{j}) & = \sum_{k =k_j}^{n_j} s_j({k}) \\
& \\
R_{j} (k_{j}) & = \sum_{k =k_j}^{n_j} r_j({k}) \\
\end{split}
\end{equation}

Recall that, when we do not refer to parameters at all, $S_{j}$ corresponds to the stopping probability of the $j^{th}$ round and is not a function of the sample drawn, but of the audit. Similarly for $R_{j}$, $R$ and $S$. (See Definitions \ref{def:stopping_j} and \ref{def:risk_j}). 

\subsection{Properties of the \Minerva and \Athena complementary cdf ratios}
\label{sec:ccdf_props}
In this section we prove interesting properties of the \Minerva and \Athena ratios that are necessary to prove that the audits are risk-limiting. 

Note that the \B $(\alpha, p)$-\Bravo stopping condition is based on $\sigma(k, p, n)$: 
\[ \sigma(k, p, n) = \frac{p^k(1-p)^{n-k}}{(\frac{1}{2})^n} \]
where the history of round size is completely captured in the total number of ballots drawn, $n$.

We prove that $\sigma(k, p, n)$ is also the likelihood ratio of winner ballots in all rounds of the \Minerva and \Athena audits, even though round sizes are not constrained in any way.  We additionally prove other interesting properties. 

\begin{theorem}
\label{thm:sigma_minerva}
For the $(\alpha, p, (n_1, n_2, \ldots, n_j, \ldots) ) $-\Minerva test, if the round schedule is pre-determined (before the audit begins), the following are true for $j=1, 2, 3, \ldots$
\begin{enumerate}
\item 
\[\frac{s_j(k_j)}{r_j(k_j)} = \sigma(k_{j}, p, n_{j})\] when $r_j(k_j)$ and $s_j(k_j)$ are defined and non-zero. 
\item $\tau_{j}(k_{j}, p, (n_1, n_2, \ldots, n_j), \alpha )$ is monotone increasing as a function of $k_{j}$. 
\item $\exists k_{min, j}(\Minerva,  (n_1, n_2, \ldots, n_j), p, \alpha)$ such that 
\[ {\mathcal{A}}(X_j) ~= Correct \Leftrightarrow K_j \geq k_{min, j}(\Minerva,  (n_1, n_2, \ldots, n_j), p, \alpha)\]
\end{enumerate}
\end{theorem}

Similarly: 

\begin{theorem}
\label{thm:sigma_athena}
For the $ (\alpha, \delta, p, (n_1, n_2, \ldots, n_j, \ldots) ) $-\Athena test, if the round schedule is pre-determined (before the audit begins), the following are true for $j=1, 2, 3, \ldots$:
\begin{enumerate}
\item 
\[\frac{s_j(k_j)}{r_j(k_j)} = \sigma(k_{j}, p, n_{j})\] when $r_j(k_j)$ and $s_j(k_j)$ are defined and non-zero. 
\item $\omega_{j}(k_{j}, p, (n_1, n_2, \ldots, n_j), \alpha, \delta)$ is monotone increasing as a function of $k_{j}$.
\item $\exists k_{min, j}(\Athena,  (n_1, n_2, \ldots, n_j), p, \alpha, \delta)$ such that 
\[ {\mathcal{A}}(X_j) ~= Correct \Leftrightarrow K_j \geq k_{min, j}(\Athena,  (n_1, n_2, \ldots, n_j), p, \alpha)\]
\end{enumerate}
\end{theorem} 

\subsection{\Minerva and \Athena are risk-limiting}
\label{sec:rlas}
Now we may state the results on the risk limiting properties of the audits. 

\begin{theorem}\label{thm:minervaRLA}
$ (\alpha, H_a, (n_1, n_2, \ldots, n_j, \ldots) ) $-\Minerva is an $\alpha$-{\em RLA} if the round schedule is pre-determined (before the audit begins). 
\end{theorem}

Exactly the same approach may be used to prove: 
\begin{theorem}\label{thm:athenaRLA}
$ (\alpha, \delta, H_a, (n_1, n_2, \ldots, n_j, \ldots) ) $-\Athena is an $\alpha$-{\em RLA} if the round schedule is pre-determined (before the audit begins). 
\end{theorem}

\subsection{Properties of \B versions of \Minerva and \Athena}
In this section we study the relationship between \B \Bravo and \Minerva with each round consisting of a single ballot draw. We make the following observation: samples satisfying the stopping condition of $(\alpha, p)$-\Bravo, performed ballot-by-ballot, are exactly those satisfying that of the $(\alpha, H_a, (1, 2, 3, \ldots, j, \ldots))$-\Minerva audit, where $H_a$ is the hypothesis that the winner's fractional tally is $p$. The p-values of the two audits, however, differ except at their values of $k_{min}$. 
 
 \begin{theorem}
 \label{thm:B2Minerva}
The \B $(\alpha, p)$-\Bravo audit stops for a sample of size $n_j$ with $k_j$ ballots for the winner, if and only if the $(\alpha, H_a, (1, 2, 3, \ldots, j, \ldots))$-\Minerva audit stops.  
 \end{theorem}
 
 \begin{corollary}
Given $\alpha, p, \delta$ such that $\delta \geq \alpha$, the \B $(\alpha, p)$-\Bravo audit stops for a sample of size $n_j$ with $k_j$ ballots for the winner, if and only if the $(\alpha, \delta, H_a, (1, 2, 3, \ldots, j, \ldots))$-\Athena audit stops.  
 \end{corollary}

\subsection{Strong RLAs}
We define a new audit property which characterizes the difference between \Minerva and \Athena. 
\begin{definition}[Strong Risk Limiting Audit ($\alpha, \delta$)-RLA]\label{def:srla}
 An audit procedure $\mathcal{A}$ is an ($\alpha, \delta$)-\textit{Strong Risk Limiting Audit} if it is a Risk Limiting Audit 
 with risk level $\alpha$ and, if, for every accepted sample, the likelihood-ratio is bounded below by $\frac{1}{\delta}$:
 \[
 \mathcal{A}(X) = correct \Rightarrow \frac{Pr[X \mid H_a]}{Pr[X \mid H_0]} \geq \frac{1}{\delta}
 \]
\end{definition}

It is easy to see that: 

\begin{lemma}
\B \Bravo, {\em end-of-round} \Bravo and {\em selection-ordered-ballots} \Bravo are $(\alpha, \alpha)$-strong RLAs.
\end{lemma}

\subsection{Efficiency}
\label{sec:efficiency}

In this section we present an efficiency result for \Minerva and \Athena (for $\delta \geq \alpha)$.  

\begin{theorem}
\label{thm:efficiency}
Given sample $X$ of size $n_j$ with $k_j$ samples for the winner, 
\[ \mathcal{A}_B(X) ~=~ Correct \Rightarrow \mathcal{A}_A(X) ~=~Correct \]
where $\mathcal{A}_B$ denotes the $(\alpha, p)$-\Bravo test and $\mathcal{A}_A$ the $(\alpha, p, (n_1, n_2, \ldots, n_j, \ldots) ) $-\Minerva test or the \\ $(\alpha, \delta, p, (n_1, n_2, \ldots, n_j, \ldots) ) $-\Athena test for $\delta \geq \alpha$ if the round schedule is pre-determined (before the audit begins). 
\end{theorem}

From this it follows that \Minerva and \Athena (for $\delta \geq \alpha)$ is each at least as efficient as the corresponding {\em end-of-round} application of \B rules. In section \ref{sec:results_first_round} we demonstrate that \Athena and \Minerva can be considerably more efficient. 

\section{Experimental Results}
\label{sec:results}
In this section we describe our experimental results. We first present our analytical results for percentiles of the \Bravo stopping condition, and compare them with those reported by Lindeman {\em et al.} \cite[Table 1]{bravo}. We then describe our estimates of first round sizes, comparing \Athena ($\delta=1$) to both {\em end-of-round} \Bravo and {\em selection-ordered-ballots} \Bravo. Finally, we present simulation results. 

\subsection{\B BRAVO Percentile Verification}
\label{sec:verification}
We used the approach described in Section \ref{sec:computing} to generate the probability distributions for \B \Bravo using various election margins to see how our estimates compared to those obtained by Lindeman {\em et al.} \cite[Table 1]{bravo}. They used $10,000$ simulations. 

Table \ref{tab:bravo_verification} presents our values. Values in parentheses are from \cite[Table 1]{bravo}, where they differ. Also listed in the table is Average Sample Number (ASN), which is computed using a standard theoretical estimate (and not using our analytical expressions, nor simulations). It provides a baseline to compare with the values for the Expected Ballots column. Some of the difference between our values and those of  \cite[Table 1]{bravo} is likely due to rounding off. Further, we notice that both our values and those of \cite[Table 1]{bravo}, when they differ from ASN, are lower than ASN. In our case, the difference is likely due to the fact that we compute our probability distributions for only up to 6ASN draws, using a finite summation to estimate the probability distributions, and we model the discrete character of the problem, which is not captured by ASN. The largest difference between our values and those of \cite[Table 1]{bravo} is $190$ ballots, corresponding to a fractional difference of 0.41 \%, in the estimate of the expected number of ballots drawn for a margin of 1\%. Our value is further from ASN. The average of the absolute value of the fractional difference between our results and those of \cite{bravo} is 0.13\%. The differences between our values and those obtained with simulations could be because $10,000$ simulations may not be sufficiently accurate at the lower margins, where most of the errors are. It could also be because our finite summation is not sufficient at low margin. 
\begin{table}
\centering
\scriptsize
\begin{tabular}{||l|r|r|r|r|r|r|r||} \hline \hline
Margin & $25^{th}$ & $50^{th}$ & $75^{th}$ & $90^{th}$ & $99^{th}$ & Expected Ballots & ASN \\ \hline \hline
0.4 & 12	& 22	& 38& 	60&	131 &	 29.47 &	 30.03 \\ 
& & & & 	&	&	 (30) &	 \\ \hline	
0.3	&23	&38	&66	&108	 &236 &	52.83& 53.25	\\ 
	&	&	&	& &&	(53) & \\ \hline	
0.2	& 49	 &84	& 149	& 244	&538	  &118.00& 118.88  \\ 	
& &	& 	&	&	  &(119) &  \\ \hline	
0.18	&77	&131 &	231	& 381	& 842 & 183.60 & 184.89 \\ 	 
	&	& &		& 	& (840) & (184) & \\ \hline	 
0.1	& 193&	332	&587	 & 974	& 2,155 & 466.47  & 469.26   \\ 
& &	& & 	&  (2,157) & (469) &  \\ \hline
0.08	& 301	& 518	& 916	& 1,520	& 3,366 & 726.95	& 730.80 \\ 
	& 	& 	& 	& 	&  &  (730)	&  \\ \hline
0.06	& 531& 	914 & 1,621 &	2,698	 & 5,976& 1,287.60 &	1,294.62  \\ 
& &  & (1,619) &		 (2,700) & (5,980) & (1,294) &	\\ \hline
	
0.04	& 1,190 & 2,051 & 	3,637 & 6,055 & 13,433  & 2,887.28  & 2,901.97  \\ 
	& (1,188) & & 	 &  (6,053) & (13,455) & (2,900) & \\ \hline 

0.02	& 4,727  &	8,161 & 	14,493  &	24,155  &	53,646 &	11,506.45 & 11,561.66 \\ 
	&  (4,725) &	(8,157) & 	(14,486) & (24,149) &	(53,640) &		(11,556) & \\ \hline

0.01	& 18,845 &	32,566  &	57,856 & 96,469 &	214,385 &	45,935.85 & 46,150.44 \\
 & (18,839) &	(32,547) &	(57,838) & (96,411) &	 (214,491) &	(46,126) & \\ \hline
\end{tabular}
\caption{Computed Estimates of \B \Bravo Stopping Probabilities. Values in parentheses are those from \cite[Table 1]{bravo} that differ.}
\label{tab:bravo_verification}
\end{table}
\subsection{First-round estimates}
\label{sec:results_first_round}
In this section we report the results of our estimates for first round sizes for 90\% stopping probability for both {\em end-of-round} \Bravo and \Athena ($\delta=1$), for the announced statewide results of the 2016 US Presidential election. Our results are presented in Table \ref{tab:distinct}. 

We constructed a table of stopping probability as a function of round size for a given margin, where the stopping probability of a round is the tail corresponding to the $k_{min}$ value for that round size. We observed that the stopping probability is not a monotone increasing function of round size. This is because, if $k_{min}$ increases with round size (it does not decrease, but it may remain the same), the stopping probability may decrease slightly. For our first-round-size computations reported in section \ref{sec:results}, we use the more conservative estimates: given a desired stopping probability $\rho$, we chose round sizes such that all larger rounds stopped with probability at least $\rho$. For small margins, smaller than $0.05$ (except the states of Michigan and New Hampshire, which had the smallest margin), we did not construct the entire table, but began looking for the values by checking if the values of $k$ with the requisite tail size satisfied the stopping condition. Finally, for the states of Michigan and New Hampshire, we approximated round size by estimating the binomial as a gaussian. 

We first examined the relationship between {\em end-of-round} \Bravo and \Athena ($\delta=1$) first-round sizes. We estimated the stopping probability and the first-round sizes for {\em end-of-round} \Bravo as described above. We use the round size beyond which the stopping probability is at least $90\%$ for both {\em end-of-round} \Bravo and \Athena, thus our round-size estimates are conservative. We scaled these estimates by the ratio of total ballots cast to the number of valid ballots in the contest between the two leading candidates, Trump and Clinton. Better estimates would result from taking into consideration every possible margin for every round size. Our approach, however, is sufficient for the purposes of a rough comparison (we have developed software for the more accurate approach; it is being tested). Of course, some of these sizes are too large for consideration in a real audit; in particular, the round size for New Hampshire is more than the number of ballots cast in the election. 

It is noteworthy that, across all margins, \Athena first round sizes are about half those of {\em end-of-round} \Bravo. We note that the number of distinct ballots drawn (thanks to Philip B. Stark for the suggestion) behaves similarly with margin, except for the smallest margins, when the number of ballots drawn is so large that the number of distinct ballots drawn differs from the number of ballot draws. 

We also estimate first round sizes for $90\%$ stopping probability for {\em selection-ordered-ballots} \Bravo by treating it as a multiple-round audit. We use the approach described in Section \ref{sec:computing}, for which our verification results were presented in Section \ref{sec:verification}. Our results are presented in Table \ref{tab:distinct_SB}. We currently omit estimates for states with margins smaller than $0.01$. In the other states, we observe that the improvement on using \Athena ($\delta=1$) is $15\%-29\%$, with greater improvements for smaller margins. Recall that, unlike {\em selection-ordered-ballots} \Bravo, \Athena does not require that the ballots be noted in selection order; sample tallies are sufficient. 

We present the same data (number of ballot draws, not number of distinct ballots) in the form of plots. Figure \ref{fig:first_round_sizes} plots the first round sizes of \Athena, {\em end-of-round} \Bravo and {\em selection-ordered-ballots} \Bravo on a log scale as a function of margin. One can observe that the \Athena round sizes are the smallest, and the {\em end-of-round} the largest. Figure \ref{fig:first_round_ratios} plots the \Athena round size as a fraction of the corresponding {\em end-of-round} \Bravo and {\em selection-ordered-ballots} \Bravo round sizes. There is a small variation with margin, with the \Athena round sizes being smaller fractions for smaller margins (that is, the improvement from using \Athena is larger for smaller margins). Note that a couple of states with the largest margins do not have the largest ratios. This is likely because the round sizes are very small, and hence a difference of a single ballot changes the ratio considerably. 


\begin{table}[h!]
\centering
\scriptsize
\begin{tabular}{||l|c|r|r|r|r|c|c||}
\hline
\hline
State & Margin & \multicolumn{2}{|c|}{EoR \Bravo} & \multicolumn{2}{|c|}{\Athena} & \multicolumn{2}{|c||}{\Athena size as a fraction} \\
 &  & \multicolumn{2}{|c|}{} & \multicolumn{2}{|c|}{} & \multicolumn{2}{|c||}{ of EoR \Bravo size} \\
& & Draws & Distinct Ballots & Draws & Distinct Ballots &  Draws & Distinct Ballots\\
\hline
Alabama & 0.2875 & 181 & 181 & 94 & 94 & 0.5193 & 0.5193 \\ 
\rowcolor{Gray}
Alaska & 0.1677 & 590 & 590 & 295 & 295 & 0.5000 & 0.5000 \\ 
Arizona & 0.0378 & 10,732 & 10,710 & 5,204 & 5,199 & 0.4849 & 0.4854 \\ 
Arkansas & 0.2857 & 187 & 187 & 96 & 96 & 0.5134 & 0.5134 \\ 
\rowcolor{Gray}
California & 0.3226 & 148 & 148 & 79 & 79 & 0.5338 & 0.5338 \\ 
Colorado & 0.0537 & 5,475 & 5,470 & 2,676 & 2,675 & 0.4888 & 0.4890 \\ 
Connecticut & 0.1428 & 748 & 748 & 374 & 374 & 0.5000 & 0.5000 \\ 
\rowcolor{Gray}
Delaware & 0.1200 & 1,057 & 1,056 & 523 & 523 & 0.4948 & 0.4953 \\ 
DistrictOfColumbia & 0.9139 & 15 & 15 & 8 & 8 & 0.5333 & 0.5333 \\ 
Florida & 0.0124 & 96,608 & 96,115 & 46,563 & 46,449 & 0.4820 & 0.4833 \\ 
\rowcolor{Gray}
Georgia & 0.0532 & 5,266 & 5,263 & 2,567 & 2,567 & 0.4875 & 0.4877 \\ 
Hawaii & 0.3488 & 128 & 128 & 68 & 68 & 0.5312 & 0.5312 \\ 
Idaho & 0.3662 & 120 & 120 & 64 & 64 & 0.5333 & 0.5333 \\ 
\rowcolor{Gray}
Illinois & 0.1804 & 474 & 474 & 242 & 242 & 0.5105 & 0.5105 \\ 
Indiana & 0.2023 & 374 & 374 & 187 & 187 & 0.5000 & 0.5000 \\ 
Iowa & 0.1013 & 1,520 & 1,520 & 753 & 753 & 0.4954 & 0.4954 \\ 
\rowcolor{Gray}
Kansas & 0.2222 & 318 & 318 & 162 & 162 & 0.5094 & 0.5094 \\ 
Kentucky & 0.3134 & 155 & 155 & 79 & 79 & 0.5097 & 0.5097 \\ 
Louisiana & 0.2034 & 365 & 365 & 182 & 182 & 0.4986 & 0.4986 \\ 
\rowcolor{Gray}
Maine & 0.0319 & 15,202 & 15,049 & 7,358 & 7,322 & 0.4840 & 0.4865 \\ 
Maryland & 0.2803 & 197 & 197 & 98 & 98 & 0.4975 & 0.4975 \\ 
Massachusetts & 0.2930 & 180 & 180 & 93 & 93 & 0.5167 & 0.5167 \\ 
\rowcolor{Gray}
Michigan & 0.0024 & 2,618,926 & 2,018,381 & 1,259,688 & 1,107,933 & 0.4810 & 0.5489 \\ 
Minnesota & 0.0166 & 56,680 & 56,139 & 27,421 & 27,294 & 0.4838 & 0.4862 \\ 
Mississippi & 0.1818 & 453 & 453 & 224 & 224 & 0.4945 & 0.4945 \\ 
\rowcolor{Gray}
Missouri & 0.1964 & 401 & 401 & 201 & 201 & 0.5012 & 0.5012 \\ 
Montana & 0.2222 & 320 & 320 & 164 & 164 & 0.5125 & 0.5125 \\ 
Nebraska & 0.2710 & 213 & 213 & 110 & 110 & 0.5164 & 0.5164 \\ 
\rowcolor{Gray}
Nevada & 0.0259 & 22,943 & 22,711 & 11,110 & 11,056 & 0.4842 & 0.4868 \\ 
NewHampshire & 0.0039 & 1,007,590 & 552,067 & 475,357 & 351,311 & 0.4718  &0.6364 \\ 
NewJersey & 0.1457 & 703 & 703 & 350 & 350 & 0.4979 & 0.4979 \\ 
\rowcolor{Gray}
NewMexico & 0.0930 & 1,888 & 1,886 & 934 & 934 & 0.4947 & 0.4952 \\ 
NewYork & 0.2354 & 272 & 272 & 140 & 140 & 0.5147 & 0.5147 \\ 
NorthCarolina & 0.0381 & 10,330 & 10,319 & 5,000 & 4,998 & 0.4840 & 0.4843 \\ 
\rowcolor{Gray}
NorthDakota & 0.3962 & 98 & 98 & 55 & 55 & 0.5612 & 0.5612 \\ 
Ohio & 0.0854 & 2,077 & 2,077 & 1,018 & 1,018 & 0.4901 & 0.4901 \\ 
Oklahoma & 0.3861 & 101 & 101 & 55 & 55 & 0.5446 & 0.5446 \\ 
\rowcolor{Gray}
Oregon & 0.1231 & 1,068 & 1,068 & 535 & 535 & 0.5009 & 0.5009 \\ 
Pennsylvania & 0.0075 & 265,245 & 259,621 & 127,792 & 126,477 & 0.4818 & 0.4872 \\ 
RhodeIsland & 0.1662 & 562 & 562 & 280 & 280 & 0.4982 & 0.4982 \\ 
\rowcolor{Gray}
SouthCarolina & 0.1492 & 683 & 683 & 344 & 344 & 0.5037 & 0.5037 \\ 
SouthDakota & 0.3194 & 154 & 154 & 79 & 79 & 0.5130 & 0.5130 \\ 
Tennessee & 0.2725 & 206 & 206 & 106 & 106 & 0.5146 & 0.5146 \\ 
\rowcolor{Gray}
Texas & 0.0943 & 1,706 & 1,706 & 833 & 833 & 0.4883 & 0.4883 \\ 
Utah & 0.2477 & 329 & 329 & 165 & 165 & 0.5015 & 0.5015 \\ 
Vermont & 0.3037 & 180 & 180 & 91 & 91 & 0.5056 & 0.5056 \\ 
\rowcolor{Gray}
Virginia & 0.0565 & 4,790 & 4,788 & 2,329 & 2,329 & 0.4862 & 0.4864 \\ 
Washington & 0.1757 & 525 & 525 & 265 & 265 & 0.5048 & 0.5048 \\ 
WestVirginia & 0.4432 & 76 & 76 & 41 & 41 & 0.5395 & 0.5395 \\ 
\rowcolor{Gray}
Wisconsin & 0.0082 & 229,503 & 220,878 & 110,622 & 108,592 & 0.4820 & 0.4916 \\ 
Wyoming & 0.5141 & 59 & 59 & 29 & 29 & 0.4915 & 0.4915 \\ \hline
\end{tabular}
\caption{Comparison of {\em end-of-round} (EoR) \Bravo and \Athena First-Round Sizes for Statewide 2016 US Presidential Contests, for $\delta=1.0$ and a stopping probability of $0.9$, contd.}
\label{tab:distinct}
\end{table} 
\clearpage
\newpage


\begin{table}[h!]
\centering
\scriptsize
\begin{tabular}{||l|c|r|r|r|r|c|c||}
\hline
\hline
State & Margin & \multicolumn{2}{|c|}{EoR \Bravo} & \multicolumn{2}{|c|}{\Athena} & \multicolumn{2}{|c||}{\Athena size as a fraction} \\
 &  & \multicolumn{2}{|c|}{} & \multicolumn{2}{|c|}{} & \multicolumn{2}{|c||}{ of SB \Bravo size} \\
& & Draws & Distinct Ballots & Draws & Distinct Ballots &  Draws & Distinct Ballots\\
\hline
Alabama & 0.2875 & 122 & 122 & 94 & 94 & 0.7705 & 0.7705 \\ 
\rowcolor{Gray}
Alaska & 0.1677 & 396 & 396 & 295 & 295 & 0.7449 & 0.7449 \\ 
Arizona & 0.0378 & 7,227 & 7,217 & 5,204 & 5,199 & 0.7201 & 0.7204 \\ 
Arkansas & 0.2857 & 128 & 128 & 96 & 96 & 0.7500 & 0.7500 \\ 
\rowcolor{Gray}
California & 0.3226 & 99 & 99 & 79 & 79 & 0.7980 & 0.7980 \\ 
Colorado & 0.0537 & 3,687 & 3,685 & 2,676 & 2,675 & 0.7258 & 0.7259 \\ 
Connecticut & 0.1428 & 502 & 502 & 374 & 374 & 0.7450 & 0.7450 \\ 
\rowcolor{Gray}
Delaware & 0.1200 & 716 & 716 & 523 & 523 & 0.7304 & 0.7304 \\ 
DistrictOfColumbia & 0.9139 & 10 & 10 & 8 & 8 & 0.8000 & 0.8000 \\ 
Florida & 0.0124 & 65,051 & 64,827 & 46,563 & 46,449 & 0.7158 & 0.7165 \\ 
\rowcolor{Gray}
Georgia & 0.0532 & 3,555 & 3,554 & 2,567 & 2,567 & 0.7221 & 0.7223 \\ 
Hawaii & 0.3488 & 86 & 86 & 68 & 68 & 0.7907 & 0.7907 \\ 
Idaho & 0.3662 & 83 & 83 & 64 & 64 & 0.7711 & 0.7711 \\ 
\rowcolor{Gray}
Illinois & 0.1804 & 318 & 318 & 242 & 242 & 0.7610 & 0.7610 \\ 
Indiana & 0.2023 & 254 & 254 & 187 & 187 & 0.7362 & 0.7362 \\ 
Iowa & 0.1013 & 1,024 & 1,024 & 753 & 753 & 0.7354 & 0.7354 \\ 
\rowcolor{Gray}
Kansas & 0.2222 & 215 & 215 & 162 & 162 & 0.7535 & 0.7535 \\ 
Kentucky & 0.3134 & 104 & 104 & 79 & 79 & 0.7596 & 0.7596 \\ 
Louisiana & 0.2034 & 247 & 247 & 182 & 182 & 0.7368 & 0.7368 \\ 
\rowcolor{Gray}
Maine & 0.0319 & 10,238 & 10,169 & 7,358 & 7,322 & 0.7187 & 0.7200 \\ 
Maryland & 0.2803 & 132 & 132 & 98 & 98 & 0.7424 & 0.7424 \\ 
Massachusetts & 0.2930 & 122 & 122 & 93 & 93 & 0.7623 & 0.7623 \\ 
\rowcolor{Gray}
Michigan & 0.0024 & - & - & 1,259,688 & 1,107,933 & - & - \\ 
Minnesota & 0.0166 & 38,185 & 37,939 & 27,421 & 27,294 & 0.7181 & 0.7194 \\ 
Mississippi & 0.1818 & 302 & 302 & 224 & 224 & 0.7417 & 0.7417 \\ 
\rowcolor{Gray}
Missouri & 0.1964 & 267 & 267 & 201 & 201 & 0.7528 & 0.7528 \\ 
Montana & 0.2222 & 217 & 217 & 164 & 164 & 0.7558 & 0.7558 \\ 
Nebraska & 0.2710 & 144 & 144 & 110 & 110 & 0.7639 & 0.7639 \\ 
\rowcolor{Gray}
Nevada & 0.0259 & 15,462 & 15,357 & 11,110 & 11,056 & 0.7185 & 0.7199 \\ 
NewHampshire & 0.0039 & - & - & 475,357 & 351,311 & - & - \\ 
NewJersey & 0.1457 & 478 & 478 & 350 & 350 & 0.7322 & 0.7322 \\ 
\rowcolor{Gray}
NewMexico & 0.0930 & 1,276 & 1,275 & 934 & 934 & 0.7320 & 0.7325 \\ 
NewYork & 0.2354 & 186 & 186 & 140 & 140 & 0.7527 & 0.7527 \\ 
NorthCarolina & 0.0381 & 6,961 & 6,956 & 5,000 & 4,998 & 0.7183 & 0.7185 \\ 
\rowcolor{Gray}
NorthDakota & 0.3962 & 70 & 70 & 55 & 55 & 0.7857 & 0.7857 \\ 
Ohio & 0.0854 & 1,403 & 1,403 & 1,018 & 1,018 & 0.7256 & 0.7256 \\ 
Oklahoma & 0.3861 & 69 & 69 & 55 & 55 & 0.7971 & 0.7971 \\ 
\rowcolor{Gray}
Oregon & 0.1231 & 724 & 724 & 535 & 535 & 0.7390 & 0.7390 \\ 
Pennsylvania & 0.0075 & - & - & 127,792 & 126,477 & - & - \\ 
RhodeIsland & 0.1662 & 382 & 382 & 280 & 280 & 0.7330 & 0.7330 \\ 
\rowcolor{Gray}
SouthCarolina & 0.1492 & 460 & 460 & 344 & 344 & 0.7478 & 0.7478 \\ 
SouthDakota & 0.3194 & 102 & 102 & 79 & 79 & 0.7745 & 0.7745 \\ 
Tennessee & 0.2725 & 138 & 138 & 106 & 106 & 0.7681 & 0.7681 \\ 
\rowcolor{Gray}
Texas & 0.0943 & 1,150 & 1,150 & 833 & 833 & 0.7243 & 0.7243 \\ 
Utah & 0.2477 & 220 & 220 & 165 & 165 & 0.7500 & 0.7500 \\ 
Vermont & 0.3037 & 122 & 122 & 91 & 91 & 0.7459 & 0.7459 \\ 
\rowcolor{Gray}
Virginia & 0.0565 & 3,229 & 3,228 & 2,329 & 2,329 & 0.7213 & 0.7215 \\ 
Washington & 0.1757 & 355 & 355 & 265 & 265 & 0.7465 & 0.7465 \\ 
WestVirginia & 0.4432 & 51 & 51 & 41 & 41 & 0.8039 & 0.8039 \\ 
\rowcolor{Gray}
Wisconsin & 0.0082 & - & - & 110,622 & 108,592 & - & - \\ 
Wyoming & 0.5141 & 40 & 40 & 29 & 29 & 0.7250 & 0.7250 \\ \hline 
\end{tabular}
\caption{Comparison of {\em selection-ordered} (SB) \Bravo and \Athena First-Round Sizes for Statewide 2016 US Presidential Contests, for $\delta=1.0$ and a stopping probability of $0.9$, contd.}
\label{tab:distinct_SB}
\end{table} 
\clearpage
\newpage

\begin{figure}[h!]
\begin{centering}
 \includegraphics[width=\columnwidth]{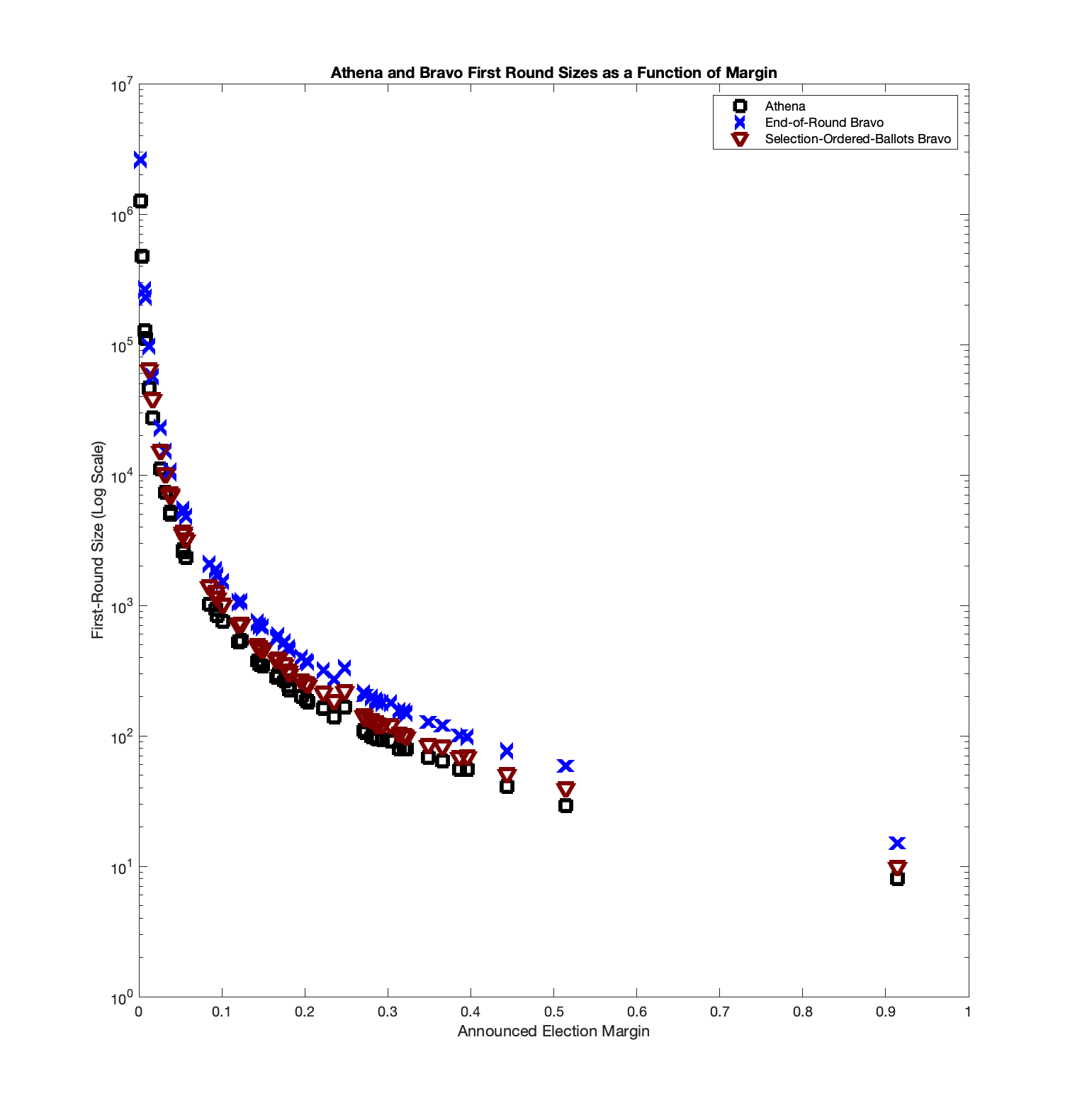}
 \caption{First-Round Sizes for 90\% stopping probability: {\em End-of-Round} \Bravo, {\em Selection-Ordered-Ballots} \Bravo and \Athena as a function of statewide margins of the 2016 US Presidential contest. }
 \label{fig:first_round_sizes}
 \end{centering}
\end{figure}
\clearpage
\newpage

\begin{figure}[h!]
\begin{centering}
 \includegraphics[width=\columnwidth]{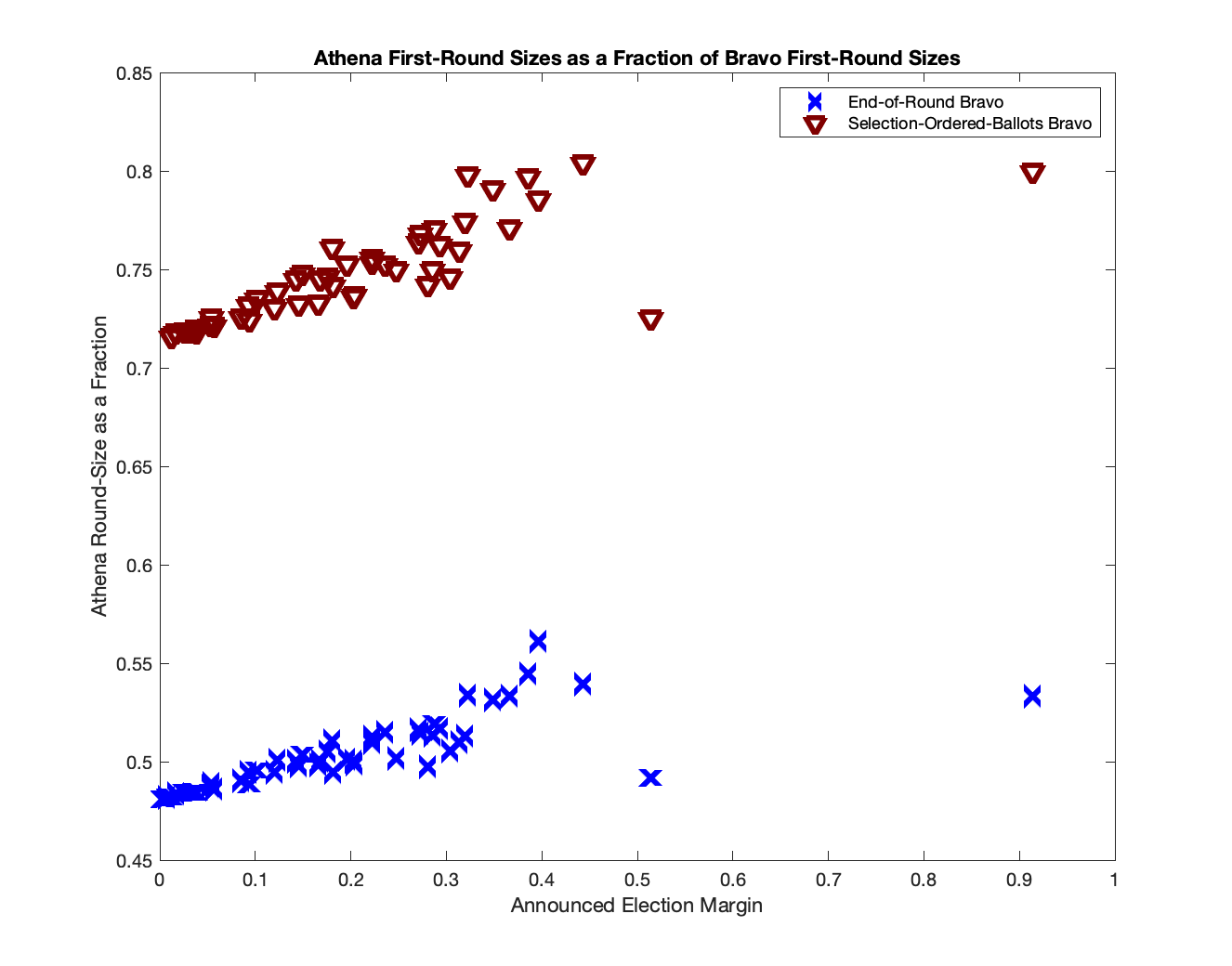}
 \caption{First-Round Sizes: \Athena first-round sizes for 90\% stopping probability as a fraction of those of {\em End-of-Round} \Bravo and {\em Selection-Ordered-Ballots}, for the statewide margins of the 2016 US Presidential contest. }
 \label{fig:first_round_ratios}
 \end{centering}
\end{figure}

\subsection{First-round Simulations}
\label{sec:sim}
We observed that the stopping conditions ($k_{min}$ values) for both \Athena ($\delta=1$) and \Minerva are identical for the \Athena first round sizes presented in Table \ref{tab:distinct}. This is because, for round sizes with large \Minerva stopping probabilities, the value of $k$ is a very good representative of the underlying distribution, and the likelihood ratio for $k \geq k_{min}$ is larger than $1$. 

We performed $100,000$ simulations of \Minerva for each of the round sizes and corresponding margins (except for a couple of the low margin states, Michigan and New Hampshire); the results are presented in Table \ref{tab:sim}. The simulations used the declared tallies, and hence included ballots that were not votes for the two main candidates.

We observed that the empirical stopping probabilities for each state were slightly larger than $90\%$. Additionally, we observed that the risk of the first round for each state was smaller than $9\%$, and hence that the stopping probability to risk ratio was larger than $\alpha^{-1} = 10$, which is as required by the stopping condition for \Minerva. Future work will include larger-scale and more complete simulations, as the risk-limited nature of the audit would need to be verified over multiple audit rounds. 
\newpage
\begin{table}[H]
\centering
\scriptsize
\begin{tabular}{|lccc|}
\hline
\hline
State & Margin & Round  Risk & Round Stopping Probability \\
\hline
Alabama & 0.2875 & 0.0776 & 0.9152 \\ 
\rowcolor{Gray}
Alaska & 0.1677 & 0.0817 & 0.9067 \\ 
Arizona & 0.0378 & 0.0888 & 0.9019 \\ 
Arkansas & 0.2857 & 0.0790 & 0.9107 \\ 
\rowcolor{Gray}
California & 0.3226 & 0.0730 & 0.9167 \\ 
Colorado & 0.0537 & 0.0876 & 0.9009\\ 
Connecticut & 0.1428 & 0.0821 & 0.9049  \\ 
\rowcolor{Gray}
Delaware & 0.1200 & 0.0844 & 0.9036 \\ 
District Of Columbia & 0.9139 & 0.0416 & 0.9478  \\ 
Florida & 0.0124 & 0.0886 & 0.9007\\ 
\rowcolor{Gray}
Georgia & 0.0532 & 0.0870 & 0.9018 \\ 
Hawaii & 0.3488 & 0.0728 & 0.9156  \\ 
Idaho & 0.3662 & 0.0761 & 0.9144 \\ 
\rowcolor{Gray}
Illinois & 0.1804 & 	0.0791 & 0.9097 \\ 
Indiana & 0.2023 & 0.0811 & 0.9083 \\ 
Iowa & 0.1013 & 0.0823 & 0.9051  \\ 
\rowcolor{Gray}
Kansas & 0.2222 & 0.0777 & 0.9074 \\ 
Kentucky & 0.3134 & 0.0748 & 0.9085 \\ 
Louisiana & 0.2034 & 0.0818 & 0.9074 \\ 
\rowcolor{Gray}
Maine & 0.0319 & 0.0896 & 0.9010 \\ 
Maryland & 0.2803 & 0.0800 & 0.9082 \\ 
Massachusetts & 0.2930 & 0.0736 & 0.9041 \\ 
\rowcolor{Gray}
Michigan & 0.0024 & - & - \\ 
Minnesota & 0.0166 & 0.0894 & 0.9008 \\ 
Mississippi & 0.1818 & 0.0836 & 0.9078 \\ 
\rowcolor{Gray}
Missouri & 0.1964 & 0.0793 & 0.9067 \\ 
Montana & 0.2222 & 0.0769 & 0.9080 \\ 
Nebraska & 0.2710 & 0.0739 & 0.9088 \\ 
\rowcolor{Gray}
Nevada & 0.0259 & 	0.0881 & 0.9006 \\ 
New Hampshire & 0.0039 & - & -  \\ 
New Jersey & 0.1457 & 0.0842 & 0.9034 \\ 
\rowcolor{Gray}
New Mexico & 0.0930 & 0.0852 & 0.9039 \\ 
New York & 0.2354 & 0.0771 & 0.9075 \\ 
North Carolina & 0.0381 & 0.0874 & 0.9018 \\ 
\rowcolor{Gray}
North Dakota & 0.3962 & 0.0705 & 0.9192  \\ 
Ohio & 0.0854 & 0.0858 & 0.9043 \\ 
Oklahoma & 0.3861 & 0.0732 & 0.9210 \\ 
\rowcolor{Gray}
Oregon & 0.1231 & 0.0830 & 0.9050\\ 
Pennsylvania & 0.0075 & 0.0896 & 0.9006 \\ 
Rhode Island & 0.1662 & 0.0811 & 0.9050 \\ 
\rowcolor{Gray}
South Carolina & 0.1492 & 0.0813 & 0.9065 \\ 
South Dakota & 0.3194 & 	0.0725 & 0.9097 \\ 
Tennessee & 0.2725 & 0.0752 & 0.9091 \\ 
\rowcolor{Gray}
Texas & 0.0943 & 	0.0861 & 0.9021 \\ 
Utah & 0.2477 & 0.0774 & 0.9084 \\ 
Vermont & 0.3037 & 0.0788 & 0.9076 \\ 
\rowcolor{Gray}
Virginia & 0.0565 & 0.0879 & 0.901 \\ 
Washington & 0.1757 & 0.0812 & 0.9079\\ 
West Virginia & 0.4432 & 0.0645 & 0.9126 \\ 
\rowcolor{Gray}
Wisconsin & 0.0082 & 0.089 & 0.9006  \\ 
Wyoming & 0.5141 & 0.0712 & 0.9039 \\ \hline 
\end{tabular}
\caption{\Minerva Simulation Results for First-Round Sizes for Statewide 2016 US Presidential Contests; stopping probability of $0.9$ and $\alpha=0.1$.}
\label{tab:sim}
\end{table}

Figures \ref{fig:sim_sprob}, \ref{fig:sim_risk} and \ref{fig:sim_ratio}  present the stopping probability, the risk and the ratio of stopping probability to risk as a function of margin for all states except DC, which has a very large margin. Larger margins have very small round sizes, and the difference of a few ballots makes greater impact. This explains the points at margins of 0.5141 (Wyoming) and 0.4432 (West Virginia). 

\begin{figure}[h!]
\begin{centering}
 \includegraphics[width=\columnwidth]{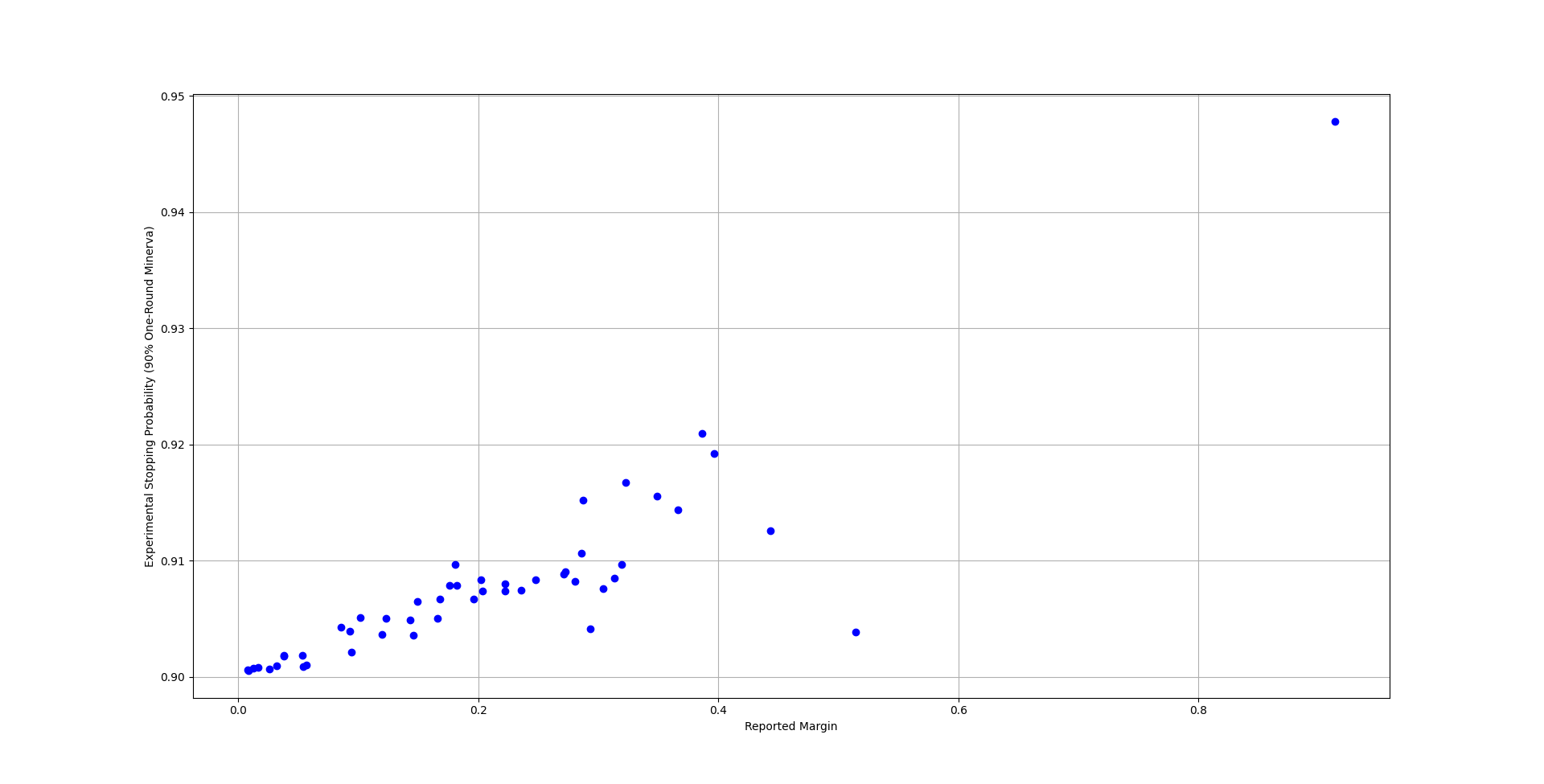}
 \caption{Simulation Results: Stopping Probability for Predicted \Athena First Rounds for the 2016 US Presidential contest, Table \ref{tab:distinct}. }
 \label{fig:sim_sprob}
 \end{centering}
\end{figure}

\begin{figure}[h!]
\begin{centering}
 \includegraphics[width=\columnwidth]{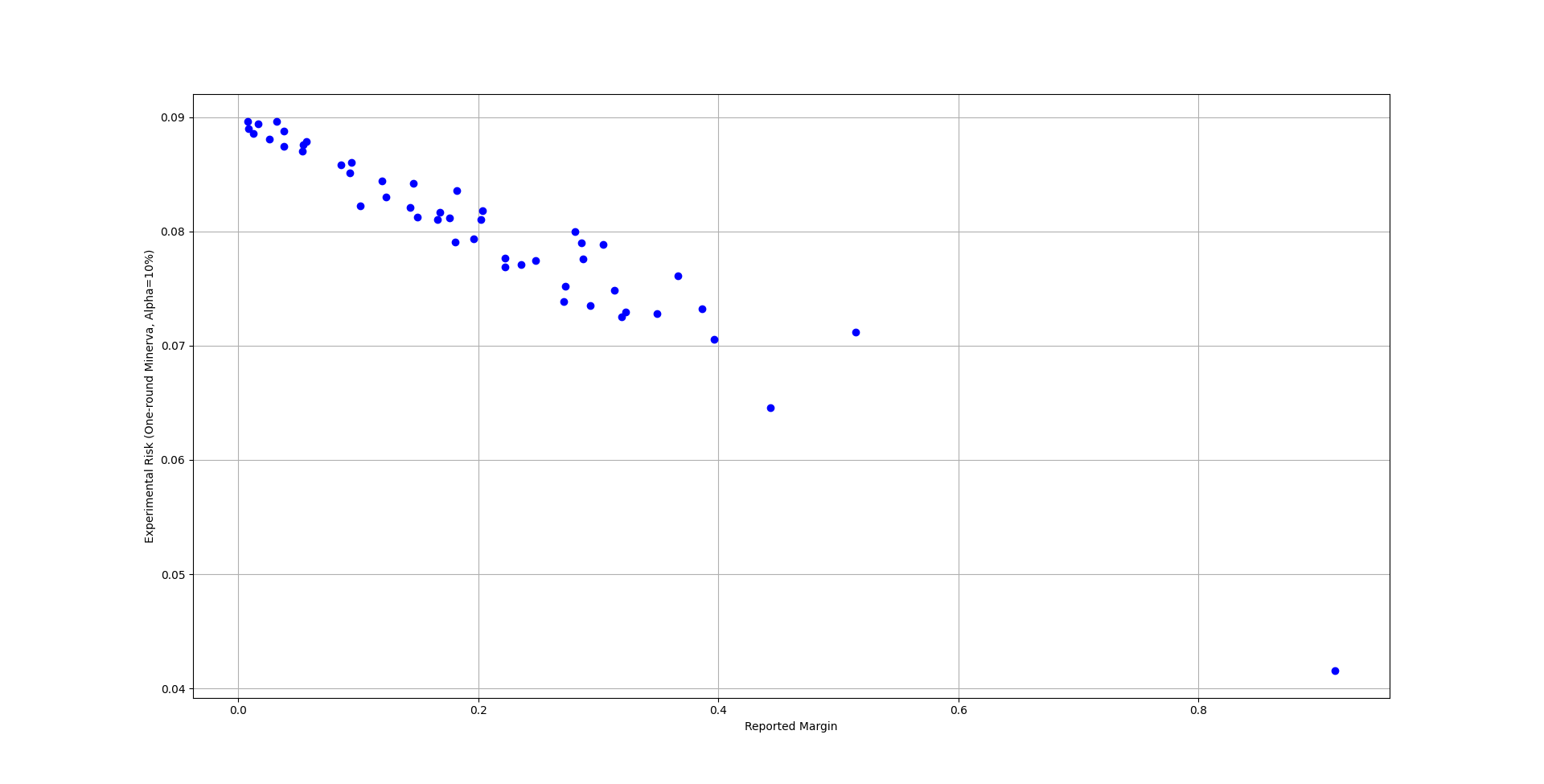}
 \caption{Simulation Results: Risk for Predicted \Athena First Rounds for the 2016 US Presidential contest, Table \ref{tab:distinct}. }
 \label{fig:sim_risk}
 \end{centering}
\end{figure}

\begin{figure}[h!]
\begin{centering}
 \includegraphics[width=\columnwidth]{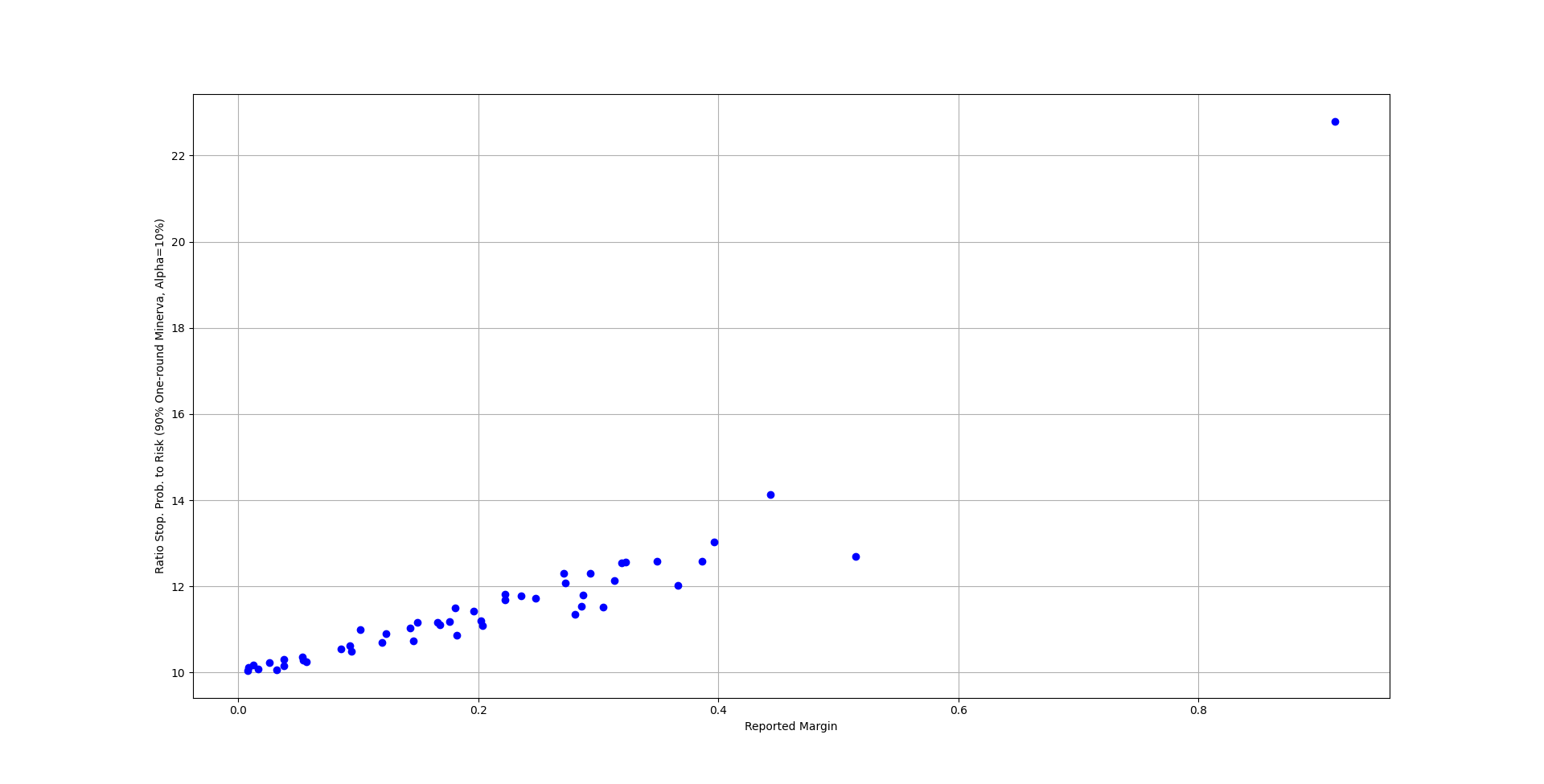}
 \caption{Simulation Results: Ratio of Stopping Probability to Risk for Predicted \Athena First Rounds for the 2016 US Presidential contest, Table \ref{tab:distinct}. }
 \label{fig:sim_ratio}
 \end{centering}
\end{figure}

\section{Conclusion}
\label{sec:conc}
We describe inefficiencies with the use of audits developed for ballot-by-ballot decisions in round-to-round procedures, such as are in use in real audits today. We propose new audits, \Minerva and \Athena, which we prove are risk-limiting and at least as efficient as audits that apply the ballot-by-ballot decision rules at the end of the round. 

We describe an approach to computing stopping probabilities and risks of audits with stopping conditions that are monotone increasing with the number of ballots for the winner in the sample. We demonstrate its accuracy in reproducing the empirically-obtained percentile values from \cite[Table 1]{bravo}, and find that the average fractional discrepancy is $0.13\%$. 

We predict first round sizes (for 90\% stopping probability) for all states in the US Presidential election of 2016 for {\em end-of-round} \Bravo and \Athena ($\delta=1$). We find that our proposed audits require half the ballots across all margins. We similarly compare first round sizes to {\em ordered-ballot-draw} \Bravo as well, finding $15-29\%$ improvements, with the larger improvements corresponding to smaller margins. We thus see that the additional effort of retaining information on ballot order, required by {\em selection-ordered-ballots} \Bravo, is not beneficial as the \Athena class of audits does not require it. 

We hope to present a third audit of the \Athena class, \Metis, which is more efficient for multiple-round audits, in a future draft of this manuscript. 

Large first-round sizes for polling audits of low margin contests should not deter election officials from performing audits. Other options exist besides those reported in Section \ref{sec:results}, which presents results for ballot polling audits only. Ballot comparison audits are far more efficient in terms of number of ballots needed for the audit; if cast vote records (CVRs) which can be efficiently matched with the corresponding paper ballot are easily available, or their creation requires less effort than the random sampling of a large number of ballots, they should be considered, especially for low margin contests. It might also be possible to perform a combination of ballot polling audits and ballot comparison audits---such as described by Ottoboni {\em et al.} in the paper on {\em SUITE} \cite{suite}---to reduce effort. 

We provide open-source software for computing probability distributions and for the \Minerva and \Athena audits, hoping it can help developers of election auditing software. We also hope our work can help election officials planning audits. 
\bibliography{audits}
\appendix
\section{Proofs}
\subsection{Preliminaries}
\label{sec:prelims}
Before we prove the Theorems, we need the following general results from basic algebra. 
\begin{lemma}
\label{thm:tail_increasing}
Given a monotone increasing sequence: $\frac{a_1}{b_1}, \frac{a_2}{b_2}, \ldots, \frac{a_n}{b_n}$, for $a_i, b_i > 0$, the sequence: 
\[ z_i = \frac{\sum_{j=i}^n a_j}{\sum_{j=i}^n b_j} \]
is also monotone increasing. 
\end{lemma}

\begin{proof}
Note that $z_i$ is a weighted average of the values of $\frac{a_j}{b_j}$ for $j \geq i$: 
\[ z_i = \sum_{j=i}^n y_j \frac{a_j}{b_j} \]
\[ y_j = \frac{b_j}{\sum _{j=i}^n b_j} > 0 \]
\[ \sum _{j=i}^n y_j = 1 \Rightarrow y_j \leq 1 \]
\[ y_j = 1 \Leftrightarrow i=j=n \]
Observe that, because $\frac{a_i}{b_i}$ is monotone increasing, 
\[ z_i \geq \frac{a_i}{b_i} \]
with equality if and only if $i=n$. 
Suppose $i < n$. Then: 
\[ z_{i+1} \geq \frac{a_{i+1}}{b_{i+1}} > \frac{a_i}{b_i} \]
\[ z_i = y_i \frac{a_i}{b_i} + (1-y_i) z_{i+1} < z_{i+1} \] 
And $z_i$ is also monotone increasing. 
\end{proof}

\begin{lemma}
\label{thm:kmin_exists}
Given a strictly monotone increasing sequence: $x_1, x_2, \ldots x_n $ and some constant $A$, 
\[ \exists i_{min}~such~that~x_i \geq A \Leftrightarrow i \geq i_{min} \]
\end{lemma}

\begin{proof}
Let $i_{min}$ be the first index for which the sequence exceeds or equals $A$. That is, let $i_{min}$ be such that 
\[ x_{i_{min}} \geq A~~ x_j < A~~1 \leq j < i_{min}\]
Because the sequence is monotone increasing, 
\[ x_i > x_{i_{min}} \geq A ~~\forall i > i_{min}\]
If no elements in the sequence exceed or equal $A$, let $i_{min} = n+1$. 
\end{proof}

\begin{lemma}
\label{thm:sigma_increasing}
Given $p, n$, with $p > \frac{1}{2}$, $\sigma(k, p, n)$ is strictly monotone increasing as a function of $k$.  
\end{lemma}

\begin{proof}
\[ p > \frac{1}{2} \Rightarrow p > 1-p \Rightarrow \frac{1-p}{p} < 1 \]
Let $0 \leq k < n$. Then: 
\[ \sigma(k, p, n) = \frac{p^k(1-p)^{n-k}}{(\frac{1}{2})^n} = \frac{1-p}{p} \cdot \frac{p^{k+1}(1-p)^{n-(k+1)}}{(\frac{1}{2})^n} = \frac{1-p}{p} \sigma(k+1, p, n) \]
Hence 
\[ \sigma(k, p, n) < \sigma(k+1, p, n) ~\forall~k~such~that~ 0 \leq k < n\]
\end{proof}

\subsection{Proofs of properties of the complementary cdf ratios}
\setcounter{theorem}{0}
\begin{theorem}
For the $(\alpha, p, (n_1, n_2, \ldots, n_j, \ldots) ) $-\Minerva test, if the round schedule is pre-determined (before the audit begins), the following are true for $j=1, 2, 3, \ldots$
\begin{enumerate}
\item 
\[\frac{s_j(k_j)}{r_j(k_j)} = \sigma(k_{j}, p, n_{j})\] when $r_j(k_j)$ and $s_j(k_j)$ are defined and non-zero. 
\item $\tau_{j}(k_{j}, p, (n_1, n_2, \ldots, n_j), \alpha )$ is strictly monotone increasing as a function of $k_{j}$. 
\item $\exists k_{min, j}(\Minerva,  (n_1, n_2, \ldots, n_j), p, \alpha)$ such that 
\[ {\mathcal{A}}(X_j) ~= Correct \Leftrightarrow K_j \geq k_{min, j}(\Minerva,  (n_1, n_2, \ldots, n_j), p, \alpha)\]
\end{enumerate}
\end{theorem}
\begin{proof}
We show this by induction. 

Consider $j=1$. 
\begin{enumerate}
\item \[ \frac{s_1(k_1)}{r_1(k_1)} = \frac{Pr[K_1 = k_{1} \mid H_a, n_1]}{Pr[K_1 = k_1 \mid tie, n_{1}]} = \sigma(k_1, p, n_1)\]
\item \[ \tau _{1}(k_1, p, n_1) =  \frac{Pr[K_1 \geq k_{1} \mid H_a, n_1]}{Pr[K_1 \geq k_{1} \mid tie, n_1]} = \frac{S_1(k_1)}{R_1(k_1)} = \frac{\sum_{k =k_1}^{k_{max,1}} s_1({k})}{\sum _{k_{max,1}}^n r_1(k)} \]
where $k_{max,j}$ is the largest possible value for $k_j$. Note that $k_{max,1}=n_1$. 

is a weighted average of $\sigma(k, p, n_1)$, and, by Lemmas \ref{thm:tail_increasing} and \ref{thm:sigma_increasing}, is strictly monotone increasing as a function of $k_1$. 
\item From Lemma \ref{thm:kmin_exists}, $\exists k_{min, 1}(\Minerva,  (n_1), p, \alpha)$ such that 
\[ \tau _{1}(k_1, p, n_1) \geq \frac{1}{\alpha} \Leftrightarrow k_1 \geq k_{min, 1}(\Minerva,  (n_1), p, \alpha) \]
which is the Minerva stopping condition. 
\end{enumerate}
Thus the theorem is true for $j=1$. 

Suppose the theorem is true for $j=m$. We will show it is true for $j=m+1$. 

From property (3) of this theorem for $j=m$, we observe that, after the stopping decision is made and before the next round is drawn, the number of winner ballots in the sample is strictly smaller than $k_{min, m}(\Minerva,  (n_1, n_2, \ldots, n_m), p, \alpha)$. The distribution on the winner votes may be modeled as $s_m^*(k_m)$ and $r_m^*(k_m)$ where: 
\begin{equation*}
s_m^*(k_m) = \left\{ \begin{array}{ll} s_m(k_m) & k < k_{min, m}(\Minerva,  (n_1, n_2, \ldots, n_m), p, \alpha) \\
0 & else\\
\end{array}
\right .
\end{equation*}
and
\begin{equation*}
r_m^*(k_m) = \left\{ \begin{array}{ll} r_m(k_m) & k < k_{min, m}(\Minerva,  (n_1, n_2, \ldots, n_m), p, \alpha) \\
0 & else\\
\end{array}
\right .
\end{equation*}
When we draw the next round of ballots with replacement, the resulting distributions on the winner ballots are convolutions:
\[ s_{m+1}(k_{m+1}) = s_m^*(k_m) \circledast binomial(k_{new,m+1}, p, n_{m+1}-n_m) \]
and
\[ r_{m+1}(k_{m+1}) = r_m^*(k_m) \circledast binomial(k_{new,m+1}, 0.5, n_{m+1}-n_m) \]
where $\circledast$ represents the convolution operator and $binomial(k, p, n)$ the probability of drawing $k$ ballots for the winner in a sample of size $n$ from a distribution with fractional tally $p$ for the winner. 
Using property (1) of this theorem for $j=m$, we see that 
\[ s_m^*(k_m) = A(k_m) p^{k_m}(1-p)^{n_m-k_m}\]
and
\[r_m^*(k_m) = A(k_m) (\frac{1}{2})^{n_m} \]
for some $A$, a function of $k$, current and previous round sizes, $p$ and $\alpha$. 

Some book keeping demonstrates that 
\[ s_{m+1}(k_{m+1}) = B(k_{m+1}) p^{k_{m+1}}(1-p)^{n_{m+1}-k_{m+1}}\]
where 
\[ B(k_{m+1}) =  A(k_m) \circledast { {n_{m+1}-n_m} \choose k_{new, m+1} } \]
And 
\[ r_{m+1}(k_{m+1}) = B(k_{m+1}) (\frac{1}{2})^{n_{m+1}}\]
which proves property (1) for $j=m+1$. Properties (2) and (3) follow for $j=m+1$ by application of Lemmas \ref{thm:tail_increasing}-\ref{thm:sigma_increasing}. 

Thus the theorem is true for all $j \geq 1$. 

\end{proof}

\begin{theorem}
For the $ (\alpha, \delta, p, (n_1, n_2, \ldots, n_j, \ldots) ) $-\Athena test, if the round schedule is pre-determined (before the audit begins), the following are true for $j=1, 2, 3, \ldots$:
\begin{enumerate}
\item 
\[\frac{s_j(k_j)}{r_j(k_j)} = \sigma(k_{j}, p, n_{j})\] when $r_j(k_j)$ and $s_j(k_j)$ are defined and non-zero. 
\item $\omega_{j}(k_{j}, p, (n_1, n_2, \ldots, n_j), \alpha, \delta)$ is strictly monotone increasing. 
\item $\exists k_{min, j}(\Athena,  (n_1, n_2, \ldots, n_j), p, \alpha, \delta)$ such that 
\[ {\mathcal{A}}(X_j) ~= Correct \Leftrightarrow k_j \geq k_{min, j}(\Athena,  (n_1, n_2, \ldots, n_j), p, \alpha)\]
\end{enumerate}
\end{theorem}
\begin{proof}
The proof proceeds exactly as that for Theorem \ref{thm:sigma_minerva}, except that there are two stopping conditions, which may be represented as: 
\[ \omega_{j}(k_{j}, p, (n_1, n_2, \ldots, n_j) ) \geq \frac{1}{\alpha} \Leftrightarrow k_j \geq k_{min, j}(\omega,  (n_1, n_2, \ldots, n_j), p, \alpha) \]
and
\[ \sigma(k_{j}, p, n_j ) \geq \frac{1}{\delta} \Leftrightarrow k_j \geq k_{min}(\Bravo,  n, p, \alpha) \]
Hence
\[ k_{min, j}(\Athena,  (n_1, n_2, \ldots, n_j), p, \alpha) = max(k_{min, j}(\omega,  (n_1, n_2, \ldots, n_j), p, \alpha), k_{min}(\Bravo,  n, p, \alpha) )\]
\end{proof}

\subsection{Proof of risk-limiting property of \Minerva}

\begin{theorem}
$ (\alpha, H_a, (n_1, n_2, \ldots, n_j,  n_{j+1}, n_{j+2}, \ldots) ) $-\Minerva is an $\alpha$-{\em RLA} if the round schedule is pre-determined (before the audit begins). 
\end{theorem}

\begin{proof}
From Definition \ref{def:risk_j} and Theorem \ref{thm:sigma_minerva}, we have 
\begin{equation*}
\begin{split}
R_j & = Pr[ K_j \geq  k_{min, j}(\Minerva,  (n_1, n_2, \ldots, n_j), p, \alpha) \mid H_0, n_{j}] \\
& \leq \alpha  \cdot Pr[ K_j \geq  k_{min, j}(\Minerva,  (n_1, n_2, \ldots, n_j), p, \alpha) \mid H_a, n_{j}] \\
& = \alpha \cdot S_j \\
\end{split}
\end{equation*}
because $k_{min, j}(\Minerva,  (n_1, n_2, \ldots, n_j), p, \alpha)$ satisfies the \Minerva stopping condition. 

Define the total stopping probability of the audit as follows:
\[ S = Pr[({\mathcal{A}}(X) ~= Correct)  \mid H_a] \]

Then, 
\begin{equation}
\label{eqn:stopping}
S = \sum _{j} S_j \leq 1
\end{equation}

The risk of the audit is defined as: 
\[ R = Pr[({\mathcal{A}}(X) ~= Correct)  \mid H_0] = \sum _{j} R_j \leq \alpha \cdot \sum_j S_j = \alpha \cdot S \leq \alpha \]
from Equation (\ref{eqn:stopping}). 
\end{proof}

\subsection{Properties of \B versions of \Minerva and \Athena}

\setcounter{theorem}{4}
\begin{theorem}
The \B $(\alpha, p)$-\Bravo audit stops for a sample of size $n_j$ with $k_j$ ballots for the winner, if and only if the $(\alpha, H_a, (1, 2, 3, \ldots, j, \ldots))$-\Minerva audit stops.  
\end{theorem}

 \begin{proof}
Consider the $j^{th}$ round of the \Minerva audit: the $j^{th}$ ballot draw. Suppose that, before the $j^{th}$ round is drawn, and after the stopping condition is tested for the $(j-1)^{th}$ round and the audit stopped if it is satisfied, $k$ is the largest value of winner ballots possible. It is strictly smaller than the corresponding  $k_{min, j-1}$, because the audit has stopped for all other values. Further, because at most one winner ballot will be drawn in the $j^{th}$ round, the largest possible number of winner ballots in the $j^{th}$ round is $k+1$. 
  
More formally, let the largest value of $k_{j-1}$ for which $s_{j-1}^*(k_{j-1}) \neq 0$ be $k$, where $s_j^*$ is as defined in the proof of Theorem \ref{thm:sigma_minerva}. Then 
\[ k < k_{min, j-1}(\Minerva,  (n_1, n_2, \ldots, n_{j-1}), p, \alpha)\]
 by the definition of  $k_{min, j-1}$, Theorem \ref{thm:sigma_minerva}.  Further, the largest value of $k_{j}$ for which $s_{j}(k_{j}) \neq 0$ is $k+1$. 
 
We now show that if the $j^{th}$ round stops at all, it will be for $k_j = k+1$ and no other values of $k_j$. 

We observe that the only way to obtain $k+1$ ballots in the $j^{th}$ round is if the existing number of winner ballots is $k$ and the new ballot drawn is for the winner. The probability is:
\[ s_{j}(k+1) = ps_{j-1}(k) \]
On the other hand, $k$ ballots arise in the $j^{th}$ round if the existing number is $k-1$ and a winner ballot is drawn, or the existing number is $k$ and the ballot drawn is not for the winner. 
\[ s_{j}(k) = (1-p) s_{j-1}(k) + ps_{j-1}(k-1)\]
Similarly: 
\[ r_{j}(k+1) = \frac{1}{2} r_{j-1}(k) \]
and 
\[ r_{j}(k) =  \frac{1}{2} r_{j-1}(k) +  \frac{1}{2} r_{j-1}(k-1)\]

If the condition is satisfied by values other than $k+1$, because $\tau$ is monotone increasing, it is satisfied by $k$: 
\[ \tau _{j} (k) = \frac{s_j(k+1) + s_j(k)}{r_j(k+1) + r_j(k) } = \frac{s_{j-1}({k}) + ps_{j-1}(k-1)}{r_{j-1}({k}) +  \frac{1}{2} r_{j-1}(k-1)} \geq \frac{1}{\alpha}\]

Thus $\tau_{j}(k)$  is a weighted average of $\sigma(k, p, j-1)$ and $\frac{p}{\frac{1}{2}}\sigma(k-1, p, j-1)$ and:
\[ \frac{p}{\frac{1}{2}}\sigma(k-1, p, j-1) =  \frac{(1-p)}{\frac{1}{2}}\sigma(k, p, j-1)  < \sigma(k, p, j-1) < \tau (k, p, j-1) < \frac{1}{\alpha} \]
as $k < k_{min, j-1}(\Minerva,  (1, 2, \ldots, j-1), p, \alpha)$. And hence, $\tau_{j}(k)$ does not pass the stopping condition. 

Thus, if $\mathcal{A}_M$ and $\mathcal{A}_B$ denote the \B \Minerva and \B \Bravo audits respectively, 
\[ \mathcal{A}_M(X_j) ~=~ Correct \Leftrightarrow \tau_j(k, p, j) \geq \frac{1}{\alpha} \Leftrightarrow \sigma(k, p, j) \geq \frac{1}{\alpha} \Leftrightarrow \mathcal{A}_B (X_j) ~=~ Correct \]

Samples that do satisfy the stopping condition have the same \Minerva and \Bravo p-values, which are otherwise not the same.  
 \end{proof}
 
\setcounter{corollary}{0}
\begin{corollary}
Given $\alpha, p, \delta$ such that $\delta \geq \alpha$, the \B $(\alpha, p)$-\Bravo audit stops for a sample of size $n_j$ with $k_j$ ballots for the winner, if and only if the $(\alpha, \delta, H_a, (1, 2, 3, \ldots, j, \ldots))$-\Athena audit stops.  
\end{corollary}
 \begin{proof}
As in Theorem \ref{thm:B2Minerva}, the $j^{th}$ audit round stops only for the largest possible number of winner votes if it does at all. Thus, it stops if and only if it satisfies the $(\alpha,p)$-\Bravo stopping condition. Additionally, the second stopping condition for \Athena is also a \Bravo condition, and is always satisfied when the first one is satisfied because $\delta \geq \alpha$. 
\end{proof}

\subsection{Proof of efficiency property of \Minerva and \Athena}

\begin{theorem}
Given sample $X$ of size $n_j$ with $k_j$ samples for the winner, 
\[ \mathcal{A}_B(X) ~=~ Correct \Rightarrow \mathcal{A}_A(X) ~=~Correct \]
where $\mathcal{A}_B$ denotes the $(\alpha, p)$-\Bravo test and $\mathcal{A}_A$ the $(\alpha, p, (n_1, n_2, \ldots, n_j,  n_{j+1}, n_{j+2}, \ldots) ) $-\Minerva test or the \\ $(\alpha, \delta, p, (n_1, n_2, \ldots, n_j,  n_{j+1}, n_{j+2}, \ldots) ) $-\Athena test for $\delta \geq \alpha$ if the round schedule is pre-determined (before the audit begins). 
\end{theorem}
\begin{proof}
Note that for a fixed election and fixed round sizes, each of $\tau$ and $\omega$, the complementary cdf stopping conditions for \Minerva and \Athena respectively, is a weighted sum of $\sigma$, the monotone increasing \Bravo stopping condition. Further, if $k$ is the number of winner ballots, the elements in the weighted sum are at least as large as $\sigma$. In fact, equality for $\tau$ occurs only when $k$ is the largest possible number of winner ballots in the round. Thus 
\[ \sigma(k_j, p, n_j) \geq \frac{1}{\alpha} \Rightarrow \tau_j(k_j, p, (n_1, n_2, \ldots, n_j) ), \omega_j(k_j, p, (n_1, n_2, \ldots, n_j, \ldots) ) \geq \frac{1}{\alpha} \]
Note that \Athena has a second condition, which is also satisfied
\[  \sigma(k_j, p, n_j) \geq \frac{1}{\alpha} \Rightarrow \sigma(k_j, p, n_j)  \geq \frac{1}{\delta} \]
\end{proof}

\end{document}